\documentclass[%
 reprint,
superscriptaddress,
 amsmath,amssymb,
 aps,
]{revtex4-2}
\usepackage{mathtools}
\usepackage{graphicx}
\usepackage{dcolumn}
\usepackage{stmaryrd}
\usepackage{bm}
\usepackage{hyperref} 
\usepackage{comment}
\usepackage{xspace}
\usepackage{amsthm}
\usepackage{amsmath}
\usepackage{nicematrix}
\usepackage{tikz}
\usepackage{adjustbox}
\usepackage{arydshln}
\theoremstyle{definition}
\newtheorem{dfn}{Definition}[section]
\newtheorem{prop}[dfn]{Proposition}
\newtheorem{lem}[dfn]{Lemma}

\newtheorem{rem}[dfn]{Remark}

\newcommand{\nameofcode}{PRISM\xspace}
\newcommand{\F}{\mathbb F}
\newcommand{\GL}{\operatorname{GL}}
\newcommand{\Sp}{\operatorname{Sp}}
\newcommand{\Sym}{\operatorname{Sym}}
\newcommand{\diag}{\operatorname{diag}}

\newcommand{\foldS}{\operatorname{S}_{\rm fd}}
\newcommand{\foldH}{\operatorname{H}_{\rm fd}}

\usepackage{braket}

\DeclareMathOperator{\im}{\operatorname{im}}

\usepackage{tikz-cd}
\usepackage[compat=0.8]{yquant}

\usetikzlibrary{shapes.geometric, arrows, positioning, fit, calc, decorations.pathreplacing}

\begin{document}

\title{Constant-rate quantum codes with low-weight stabilizers and full logical Clifford actions via transversal and fold-transversal gates}

\author{Takaya Matsuura}
\thanks{These authors contributed equally to this work.}
 \affiliation{OptQC Corp., 3-28-13 Nishi-Ikebukuro, Toshima-ku, Tokyo 171-0021, Japan}
 \affiliation{RIKEN Center for Quantum Computing (RQC), Hirosawa 2-1, Wako, Saitama 351-0198, Japan}
\author{Yohji Chin}%
\thanks{These authors contributed equally to this work.}
\affiliation{OptQC Corp., 3-28-13 Nishi-Ikebukuro, Toshima-ku, Tokyo 171-0021, Japan}
\affiliation{Department of Applied Physics, Graduate School of Engineering,
The University of Tokyo, 7-3-1 Hongo, Bunkyo-ku, Tokyo 113-8656, Japan}
\author{Shohei Kiryu}
\affiliation{OptQC Corp., 3-28-13 Nishi-Ikebukuro, Toshima-ku, Tokyo 171-0021, Japan}
\author{Kosuke Fukui}
\affiliation{OptQC Corp., 3-28-13 Nishi-Ikebukuro, Toshima-ku, Tokyo 171-0021, Japan}

\date{\today}

\begin{abstract}
Low-space-overhead fault-tolerant quantum computation requires not only high-rate quantum error-correcting codes but also space-efficient implementations of logical operations. Transversal and fold-transversal gates are promising since they limit error propagation and require no additional qubits. However, the logical operations they enable are typically restricted, and a central challenge is to construct codes that combine a complete set of logical Clifford gates with favorable code parameters. In this work, we introduce a family of quantum codes with an asymptotically constant encoding rate and sublogarithmically growing stabilizer weight, while supporting the entire logical Clifford group using only transversal and fold-transversal gates. Our construction is based on classical codes whose code spaces are absolutely irreducible Steinberg modules of their Tanner-graph automorphism groups.
Taking hypergraph products of these classical codes yields quantum codes for which all logical Clifford operations can be synthesized from a fixed set of transversal and fold-transversal gates. Moreover, the slow growth of the stabilizer weight enables a high error-correcting performance in small instances. These results provide a path toward fault-tolerant quantum computation with low space overhead using transversal and fold-transversal gates.
\end{abstract}

\maketitle

\section{Introduction}

Low-space-overhead fault-tolerant quantum computation requires quantum error-correcting codes that encode $k$ logical qubits into $n$ physical qubits at a high rate $k/n$, while maintaining a sufficiently large code distance $d$~\cite{gottesman2014, Tamiya2026, yamasaki2024time}. Favorable code parameters alone, however, are not enough; fault-tolerant logical operations must also be implemented without consuming many ancillary qubits. Among several approaches, two stand out as particularly promising. The first performs logical operations through fault-tolerant measurements. 
In particular, extensive research has been conducted on a protocol called surgery tailored to measure logical Pauli operators of quantum low-density parity-check (qLDPC) codes with additional ancillary qubits~\cite{Cohen2022, Cowtan2024,cross2025,he2025}. Together with high-fidelity magic state preparation, such measurements enable universal Pauli-based computation~\cite{Litinski2019}.
While surgery offers several advantages, such as broad applicability and flexibility in connecting different types of qLDPC codes,  its fault tolerance depends on a constant stabilizer weight. It also typically incurs repeated syndrome-extraction rounds and additional decoding costs. 
The second approach uses transversal operations~\cite{Nielsen2010}, which are not restricted to qLDPC codes and require no additional space overhead. Furthermore, for suitable codes, the time overhead may also be maintained by solving the decoding problem dynamically~\cite{Zhou2025}, which makes this approach particularly attractive.

Transversal gates act independently on individual physical qubits, or on corresponding qubits in different code blocks. Consequently, a fault at one location does not spread to multiple qubits within the same block, making these gates intrinsically fault-tolerant. 
Unfortunately, transversal gates alone cannot form a universal gate set~\cite{Eastin2009}, and thus one needs to combine them with other fault-tolerant primitives. An established protocol is to combine transversal Clifford gates with magic-state preparation~\cite{Bravyi2005} and teleportation~\cite{gottesman2014}.
For a code block encoding multiple logical qubits, however, transversal gates alone cannot even realize the full logical Clifford group~\cite{Chakraborty2026}. This lack of addressability is, in fact, a central obstacle for high-rate codes~\cite{Guyot2026}.

Fold-transversal gates enlarge the available logical gate set while retaining the essential advantages of transversal gates~\cite{Breuckmann2024}. They exploit an involutive duality that exchanges the $X$- and $Z$-structures of a CSS code and pairs its physical qubits. A logical operation can then be implemented by a layer of one- and two-qubit gates, together with a permutation of the physical qubits. These operations preserve bounded error propagation and require no ancillary space.
The idea first appeared in folded surface codes~\cite{Moussa2016} and was subsequently generalized through the framework of $ZX$-dualities~\cite{Breuckmann2024}. For the $\llbracket 30,8,3\rrbracket$ Bring code, for example, fold-transversal operations generate the full Clifford group on a four-logical-qubit subspace. This example demonstrates their power, but does not yield a scalable family with a growing number of protected logical qubits.

Within this framework, code automorphisms form an especially important class of fold-transversal operations~\cite{berthusen2025automorphismgadgetshomologicalproduct}. A code automorphism permutes physical qubits while preserving the stabilizer group and can thereby induce a nontrivial transformation on the logical space. The recently introduced phantom codes use such permutations to implement addressable CNOT gates between logical qubits in the same code block~\cite{koh2026, morris2026, He2026, Mao2026}.
The name ``phantom'' reflects the fact that these permutations can be absorbed into classical relabeling of physical-qubit indices and thus need not correspond to physical gate operations. Phantomness supplies the complete logical CNOT subgroup, but not, by itself, the full Clifford group.
For the phantom quantum Reed--Muller family, which has parameters $\llbracket n,k=\Theta(\log n),d=\sqrt{n}\rrbracket$ at the balanced point, the full Clifford action is obtained by combining permutation CNOTs and fold-diagonal phase gates with targeted phase-state injection and teleported Hadamard gates, but this implementation thus requires ancillary state preparation and measurement, rather than using only transversal and fold-transversal gates with no additional space overhead.
Moreover, automorphism-group bounds show that binary phantom codes of distance at least two obey $k\leq \log_2(n+1)$, apart from a special $k=4$ case \cite{morris2026}. Phantom symmetry therefore cannot by itself produce a high-rate stabilizer-code family.

Several other families of CSS codes support the full logical Clifford group through transversal and fold-transversal gates~\cite{Malcolm2026,tansuwannont2026,holmes2026}. Among these, the quantum Reed--Muller family studied in Ref.~\cite{tansuwannont2026} has code parameters $\left\llbracket
n,\Theta\bigl(\frac{n}{\sqrt{\log n}}\bigr),\sqrt n \right\rrbracket$ and hence the highest encoding rate $k/n$ among the previously known families, although the rate still vanishes asymptotically.
Moreover, even with an optimized generating set, its stabilizer generators have weight $2\sqrt{n}$. Its large logical dimension and $\sqrt{n}$ code distance therefore come at the cost of high-density stabilizers.

The quantum logic codes~\cite{holmes2026} provide another route to a complete and individually addressable logical Clifford gate set. Starting from a small $\llbracket n_0,2,d_0\rrbracket$ core, the construction combines $r$-fold tiling with $\ell$ levels of concatenation by the Steane code, resulting in parameters $\left\llbracket n_0r7^\ell,2r,d_0 3^\ell \right\rrbracket $. Every member admits a constant-depth Clifford gate set built from fold-transversal gates in a broader sense; it allows disjoint one- and two-qubit physical gates in a code block that individually preserve the codespace and is not necessarily based on the $ZX$-duality. Together, these operations realize individually addressable logical $S_i$, $\sqrt{X_i}$, and $\mathrm{C}Z_{ij}$ gates. However, this family cannot simultaneously achieve a constant encoding rate and a growing distance, and its stabilizer weights are also large.

The SHYPS family~\cite{Malcolm2026}---Subsystem Hypergraph Product Simplex codes---is particularly attractive since it supplies fold-transversal logical Clifford computation with constant-weight gauge checks. Automorphisms of the underlying simplex codes, transversal inter-block CNOTs, and Hadamard- and phase-type fold-transversal gates generate the full logical Clifford group with efficient gate depth. Nevertheless, its parameters $\llbracket n,\Theta((\log n)^2),\Theta(\sqrt{n})\rrbracket$ imply a vanishing encoding rate. Moreover, although SHYPS codes have weight-three gauge checks, their stabilizer weights are $\Theta(\sqrt{n})$.
Thus, they are not low-weight stabilizer codes, where we use the term \emph{low weight} in a weaker sense, allowing the weight to grow polylogarithmically in the code size $n$ and the distance $d$ rather than remain constant.
These high-weight stabilizers may complicate the decoding or degrade the error-correction performance, as suggested by the use of higher-order localized-statistics decoding (LSD) for larger code sizes~\cite{Malcolm2026}.

Many additional finite-size CSS codes have full logical Clifford groups generated by fold-transversal gates, and several further families have been conjectured~\cite{albert2026}.
None, however, has been proved to simultaneously achieve a constant encoding rate, a growing distance, slowly growing stabilizer weight, and a complete logical Clifford action generated entirely by transversal and fold-transversal gates. Establishing such a family would combine the space efficiency of high-rate codes with a rich native logical gate set that requires no ancillas.

In this work, we construct a family of quantum CSS codes that occupies precisely this parameter regime. Our construction contains a family of codes with parameters $\llbracket n,\Theta(n), n^{o(1)}\rrbracket$, where the distance grows subpolynomially with $n$, and the stabilizer weight $w$ has a sublogarithmic growth, $w=(\log n)^{1/2-o(1)}$.
To our knowledge, this is the first construction that simultaneously achieves constant rate, growing distance, sublogarithmic growth of stabilizer weight, and a complete logical Clifford action generated entirely by transversal and fold-transversal operations. In particular, it is the first known non-vanishing-rate family with growing distance to possess this logical-gate property. 

A key feature of our construction is its use of arbitrary inter-block logical CNOT operations. Unlike phantom codes~\cite{koh2026,morris2026, He2026, Mao2026}, it does not rely solely on in-block addressable logical CNOT gates and thus avoids the corresponding obstruction to high-rate families~\cite{morris2026}. Our code family thus requires at least two code blocks to generate the full logical Clifford operations.
In addition, the improved code parameters come with a tradeoff: synthesizing an arbitrary logical Clifford operation from the primitive transversal and fold-transversal gates generally requires greater time overhead than in existing constructions~~\cite{koh2026, morris2026, He2026, Mao2026, Malcolm2026, tansuwannont2026, holmes2026}.

Our construction begins with a family of classical codes called \emph{Steinberg codes}~\cite{SMITH1995199}. These codes are derived from finite groups of Lie type and their Steinberg representations~\cite{Carter1972}. The canonical pair of subgroups, called the $BN$-pair, supplies the symmetric structure and incidence relations used to define these codes. At the same time, the coordinate permutations that preserve the parity-check matrix induce the Steinberg representation, making the code space a Steinberg module. Importantly, a suitable family has a constant encoding rate, subpolynomially growing distance, and sublogarithmically growing row and column weights.  We then apply the hypergraph-product (HGP) construction~\cite{Tillich2014} to two copies of a Steinberg code. 
The resulting CSS codes inherit both the constant encoding rate and the automorphisms inducing the Steinberg representation.
Combined with the $ZX$-duality~\cite{Breuckmann2024}, which HGP codes are naturally endowed with~\cite{Quintavalle_2023}, these CSS codes have a rich set of fold-transversal gates. We show that these fold-transversal gates, together with inter-block transversal CNOT gates, can generate all logical Clifford gates whenever the Steinberg representation is absolutely irreducible, which is abundant. We therefore name our code the {\bf \nameofcode} code after {\bf PR}oduct code from {\bf I}rreducible {\bf S}teinberg {\bf M}odules.

The \nameofcode codes are not qLDPC codes in the strict asymptotic sense, since their stabilizer weight grows with the block length. Nevertheless, that growth is only $(\log n)^{1/2-o(1)}$, in sharp contrast to that of existing constructions whose number of logical qubits grows polynomially with the number of physical qubits. For practical code size below $10^{18}$ physical qubits, the stabilizer weights are $\lesssim10$, comparable to those of state-of-the-art qLDPC codes. Our memory simulations further show high error-correcting performance for small instances compared to the rotated surface-code blocks with the same number of logical qubits and a similar code distance. The construction thus demonstrates that sublogarithmic stabilizer weight can coexist with a complete native Clifford action generated by code symmetries. More broadly, the construction establishes finite-group representation theory, and the Steinberg representation in particular, as a mechanism for jointly engineering quantum-code parameters and fault-tolerant logical gates.

The paper is organized as follows. In Sec.~\ref{sec:preliminaries}, we summarize the mathematical tools needed for our code construction and the logical Clifford-group generation.
In Sec.~\ref{sec:construction}, we construct the \nameofcode code family and provide a bottom-up construction of full logical Clifford gates. We also discuss the overhead of the Clifford-gate generation.
In Sec.~\ref{sec:numerical-study}, we perform memory simulations for small instances of the \nameofcode codes, and finally in Sec.~\ref{sec:discussion}, we wrap up our work with several open questions.

\section{Preliminaries}\label{sec:preliminaries}
\subsection{Classical codes for our construction}
For our quantum code construction, we use the binary top-dimensional cycle code studied in Ref.~\cite{SMITH1995199} as the base code of the HGP~\cite{Tillich2014}. Since this cycle code affords the Steinberg module, we call it the \emph{Steinberg code}~\cite{SMITH1995199}. We review its construction and basic properties below, restricting attention to Chevalley groups, i.e., a subclass of finite groups of Lie type. This restriction keeps the relation between the defining field, the Weyl group, and the local incidence parameters particularly transparent.

\subsubsection{Chamber system and Steinberg module}
Let $q$ be a prime power, $\F_q$ be a finite field, and $G$ be a nontrivial Chevalley group of rank $r\geq 1$ over $\F_q$. The group $G$ has a pair $(B,N)$ of subgroups, called a \emph{Tits system} or a \emph{$BN$-pair}, with the following properties~\cite{Abramenko2008}:
\begin{itemize}
    \item The group $G$ is generated by $B$ and $N$.
    \item The intersection $T\coloneqq B\cap N$ is a normal subgroup of $N$.
    \item The quotient $W\coloneqq N/T$, called the Weyl group, is generated by a set of involutions $\mathcal{S}=\{s_1,\ldots,s_r\}$. If $\dot{w},\dot{s}\in N$ are representatives of $w\in W$ and $s\in\mathcal{S}$, respectively, then
    \begin{align}
        (B\dot{s}B)(B\dot{w}B)
        \subseteq B\dot{s}\dot{w}B\cup B\dot{w}B,
        \label{eq:nontrivial_weyl_relation}
    \end{align}
    and
    \begin{equation}
        \dot{s}B\dot{s}^{-1}\not\subseteq B.
    \end{equation}
\end{itemize}
We henceforth fix representatives in $N$ and suppress the dots when no confusion can arise. The Bruhat decomposition is the disjoint double-coset decomposition
\begin{equation}
    G=\bigsqcup_{w\in W}BwB.
    \label{eq:bruhat_decomposition}
\end{equation}
The pair $(W,\mathcal{S})$ is a finite Coxeter system. Its length function $\ell: W\to\mathbb{Z}_{\geq 0}$ is defined by letting $\ell(w)$ be the minimum number of generators from $\mathcal{S}$ needed to express $w$.

The $BN$-pair induces a simplicial complex $\Delta$, called the spherical Tits building of $G$. The building has dimension $r-1$ as a simplicial complex. We do not need its full simplicial description; for the code construction, it suffices to describe its chambers, its panels, and their incidence. A \emph{chamber} is a highest-dimensional simplex, identified with a left coset $gB\in G/B$. For $s\in\mathcal{S}$, the corresponding minimal parabolic subgroup is
\begin{equation}
    P_s\coloneqq B\sqcup BsB. \label{eq:parabolic_subgroup}
\end{equation}
That $P_s$ is a subgroup is a standard consequence of the $BN$-pair axioms, including Eq.~\eqref{eq:nontrivial_weyl_relation}. An $s$-type \emph{panel}, i.e., a codimension-one simplex associated with $s\in{\cal S}$, is identified with a coset $gP_s\in G/P_s$. A chamber $gB$ contains the panel $hP_s$ precisely when
\begin{equation}
    gB\in hP_s/B,
    \qquad
    hP_s/B\coloneqq\{h'B\in G/B:h'\in hP_s\},
\end{equation}
or, equivalently, when $gP_s=hP_s$. Two distinct chambers $gB$ and $hB$ are $s$-adjacent when they contain the same $s$-type panel; equivalently,
\begin{equation}
    gP_s=hP_s
    \quad\Longleftrightarrow\quad
    g^{-1}h\in P_s.
\end{equation}

Let $\F_2[X]$ denote the $\F_2$-vector space with basis $X$. The chamber--panel incidence relation gives the top boundary map in the augmented simplicial chain complex of $\Delta$,
\begin{align}
    \partial:
    \F_2[G/B]
    &\longrightarrow
    \bigoplus_{s\in\mathcal{S}}\F_2[G/P_s],
    \nonumber\\
    gB&\longmapsto \sum_{s\in\mathcal{S}}gP_s.
    \label{eq:steinberg_boundary}
\end{align}
Here the coefficient field is $\F_2$ because we are interested in a binary code, and the usual orientation signs disappear in characteristic two. More explicitly,
\begin{equation}
    \F_2[G/B]\cong C_{r-1}(\Delta;\F_2),
    \qquad
    \bigoplus_{s\in\mathcal{S}}\F_2[G/P_s]
    \cong C_{r-2}(\Delta;\F_2).
\end{equation}
Since $\Delta$ has no $r$-dimensional simplices, its top reduced homology is
\begin{equation}
    \widetilde{H}_{r-1}(\Delta;\F_2)=\ker\partial.
    \label{eq:top_homology_kernel}
\end{equation}

The chambers of the fundamental apartment $A_0$ are
\begin{equation}
    \{wB:w\in W\}.
\end{equation}
These cosets are distinct: if $wB=w'B$ for representatives $w,w'\in N$, then $w^{-1}w'\in B\cap N=T$, so $w$ and $w'$ represent the same element of $W$. Each panel of $A_0$ lies in exactly two chambers~\cite{Abramenko2008}; more concretely, $wB$ and $wsB$ share an $s$-type panel since $ws P_s=wP_s$. Therefore, the \emph{apartment chain}
\begin{equation}
    a_0\coloneqq\sum_{w\in W}wB
\end{equation}
is a cycle, i.e., every panel occurs twice in $\partial a_0$ and hence cancels over $\F_2$. Every apartment is a $G$-translate of $A_0$, and its apartment cycle has the form
\begin{equation}
    a_g\coloneqq g a_0=\sum_{w\in W}gwB,
    \qquad g\in G.
\end{equation}
The Solomon--Tits theorem states, in particular, that the top homology is generated by apartment classes~\cite{Abramenko2008}. Thus,
\begin{equation}
    \ker\partial
    =\widetilde{H}_{r-1}(\Delta;\F_2)
    =\operatorname{span}_{\F_2}\{a_g:g\in G\}.
    \label{eq:apartment_span}
\end{equation}
Note that $\{a_g:g\in G\}$ are linearly dependent in general (See Prop.~\ref{prop:Steinberg_code_parameters}).

This cycle space is the Steinberg module and is denoted by
\begin{equation}
    \operatorname{St}_G(\F_2)
    \coloneqq \widetilde{H}_{r-1}(\Delta;\F_2).
\end{equation}
Indeed, left multiplication by any $g'\in G$ acts on chambers and panels as
\begin{equation}
    gB\longmapsto g'gB,
    \qquad
    gP_s\longmapsto g'gP_s,
\end{equation}
and commutes with the boundary map $\partial$. It therefore preserves $\ker\partial$ and gives the Steinberg representation of $G$ on $\operatorname{St}_G(\F_2)$~\cite{Steinberg1957}.

An important property for our later application is absolute irreducibility. 
\begin{dfn}[Absolutely irreducible]
    A representation on a vector space $V$ over $\F_2$ is \emph{absolutely irreducible} if $V\otimes_{\F_2}\overline{\F}_2$ remains irreducible, where $\overline{\F}_2$ denotes the algebraic closure of $\F_2$.
    For an irreducible finite-field representation $\pi$, this is equivalent to the statement that the only matrices commuting with every $\pi(g)$ are scalar matrices. 
    This in turn is equivalent to the statement that the set of matrices spanned by $\pi(g)$ generates the entire set of matrices $M_n(\F_2)$, where $n = \dim_{\F_2}(V)$. 
\end{dfn}

\begin{prop}[Theorems~2~(iii) and~3 of Ref.~\cite{Steinberg1957}] 
    The binary Steinberg module $\operatorname{St}_G(\F_2)$ is absolutely irreducible if and only if $[G:B]$ is odd. For a nontrivial Chevalley group over $\F_q$, this is equivalent to $q=2^f$.
\end{prop}

\subsubsection{Steinberg code} \label{sec:steinberg_code}
Following Ref.~\cite{SMITH1995199}, we regard the binary top-dimensional homology of the building as a linear code. Identify the chambers $G/B$ with the standard basis of $\F_2^{[G:B]}$ and, for every $s\in\mathcal{S}$, identify the $s$-type panels $G/P_s$ with the standard basis of the corresponding summand of $\bigoplus_{s\in\mathcal{S}}\F_2^{[G:P_s]}$. The matrix of $\partial$ in these bases is the binary chamber--panel incidence matrix $H$, whose entries are
\begin{equation}
    H_{hP_s,gB}
    =
    \begin{cases}
        1,&gP_s=hP_s,\\
        0,&gP_s\neq hP_s.
    \end{cases}
    \label{eq:steinberg_parity_check}
\end{equation}
We use $H$ as a parity-check matrix and define the Steinberg code by
\begin{equation}
    C_{\mathrm{St}}(G)
    \coloneqq\ker H
    =\ker\partial
    =\operatorname{St}_G(\F_2).
    \label{eq:steinberg_code}
\end{equation}
Left multiplication by $G$ simultaneously permutes the chamber coordinates and the panel checks while preserving incidence. It therefore gives a homomorphism from $G$ to the automorphism group of the Tanner graph of $H$, which we call \emph{Tanner-graph automorphism}.
\begin{dfn}[Tanner-graph automorphism]
    The Tanner-graph automorphism ${\rm Aut}(H)$ for a parity-check matrix $H$ is a group with pairs $\{(P_g, Q_g):g\in{\rm Aut}(H)\}$ of permutation representation such that $H P_g = Q_g H$.
\end{dfn}
The Tanner-graph automorphism of the Steinberg code $C_{\mathrm{St}}(G)$ is therefore a transitive coordinate-permutation action of $G$. This symmetry is important in the later application. The homomorphism need not be faithful; for example, central elements may act trivially on the building.

In Ref.~\cite{SMITH1995199}, the parameters of this code are determined as follows.
\begin{prop}[Proposition~6 of Ref.~\cite{SMITH1995199} \label{prop:Steinberg_code_parameters}]
    Let $G$ be a Chevalley group over $\F_q$ with $BN$-pair $(B,N)$ and Coxeter system $(W,\mathcal{S})$. Let $w_0$ be the longest element of $W$. The associated Steinberg code $C_{\mathrm{St}}(G)$ has parameters
    \begin{align}
        n&=\sum_{w\in W}q^{\ell(w)},\\
        k&=q^{\ell(w_0)}=\max_{w\in W}q^{\ell(w)},\\
        d&=|W|.
    \end{align}
\end{prop}
In particular, the first inequality follows directly from $|BwB/B|=q^{\ell(w)}$~\cite{Carter1972} and the Bruhat decomposition Eq.~\eqref{eq:bruhat_decomposition}~\cite{Hiss2011}.
Note that the cardinality of the Borel subgroup $B$ itself can be given for a Chevalley group by~\cite{Carter1972} 
\begin{equation}
    |B| = q^{\ell(w_0)}|T|, \qquad  |T| = \Theta(q^r),\label{eq:cardinality_Borel}
\end{equation}
where $T$ is as defined in $BN$-pair.

The sparsity of the parity-check matrix follows directly from the local incidence structure.
\begin{prop}
    Let $G$ be a Chevalley group over $\F_q$ with Coxeter system $(W,\mathcal{S})$. The column weight $w_{\mathrm{col}}$ and row weight $w_{\mathrm{row}}$ of the parity-check matrix $H$ of the associated Steinberg code are, respectively,
    \begin{equation}
        w_{\mathrm{col}}=|\mathcal{S}|=r,
        \qquad
        w_{\mathrm{row}}=q+1.
    \end{equation}
    Moreover, the number of rows of $H$ is
    \begin{equation}
        m_{\mathrm{chk}}
        =\sum_{s\in\mathcal{S}}[G:P_s]
        =\frac{r n}{q+1}.
        \label{eq:number_of_checks}
    \end{equation}
\end{prop}
\begin{proof}
    Each chamber contains exactly one panel of each type $s\in\mathcal{S}$. A chamber labels a bit node and a panel labels a check node, so every column of $H$ contains exactly $r=|\mathcal{S}|$ ones.

    A row indexed by an $s$-type panel $gP_s$ has one nonzero entry for each chamber containing that panel. These chambers are the elements of $gP_s/B$, and therefore the row weight is
    \begin{equation}
        |gP_s/B|=[P_s:B]=q+1,
    \end{equation}
    where the last equality follows from Eq.~\eqref{eq:parabolic_subgroup} and the fact that $|BsB/B|=q^{\ell(s)}=q$. Finally,
    \begin{equation}
        [G:P_s]=\frac{[G:B]}{[P_s:B]}=\frac{n}{q+1}
    \end{equation}
    for every $s\in\mathcal{S}$. Summing over the $r$ panel types proves Eq.~\eqref{eq:number_of_checks}.
\end{proof}

Although $H$ generally contains redundant checks, its rank is fixed by the code dimension:
\begin{equation}
    \operatorname{rank}H=n-k.
\end{equation}
Consequently, this canonically defined $H$ has a redundancy. This redundancy is unavoidable for $G$ to act transitively on the chain $C_{r-1}(\Delta;\F_2)\xrightarrow{\partial}C_{r-2}(\Delta;\F_2)$.

The row-weight formula also gives a coding-theoretic view of the reducibility for odd $q$. In that case $q+1$ is even, so the all-one chamber vector satisfies $H\boldsymbol{1}=0$. It is fixed by the coordinate-permutation action of $G$ and hence spans a nonzero proper trivial submodule of $C_{\mathrm{St}}(G)$.

\paragraph{Example: $\operatorname{SL}_3(q)$.} 
Let $G=\operatorname{SL}_3(q)$. Its building is known to be the projective plane $\operatorname{PG}(2,q)$~\cite{Abramenko2008}. Chambers are incident point--line pairs, while panels are points or lines. Thus the chamber coordinates are the edges of the point--line incidence graph, the panel checks are its vertices, and $H$ is the vertex--edge incidence matrix of that graph. Algebraically, the Borel subgroup $B$ corresponds to upper-triangular matrices, and $N$ corresponds to the normalizer of the diagonal subgroup $T$ in $B$. The Weyl group $W=N/T$ is isomorphic to the permutation group $S_3$, which is generated by two involutions, i.e., $S_3=\langle s_1,s_2|s_1^2=s_2^2=(s_1s_2)^3=1\rangle$.
Since $W\cong S_3$ has a length generating function given by
\begin{equation}
    \sum_{w\in W}q^{\ell(w)}=1+2q+2q^2+q^3
    =(q+1)(q^2+q+1),
\end{equation}
the code parameters are
\begin{equation}
    [n,k,d]
    =\bigl[(q+1)(q^2+q+1),\,q^3,\,6\bigr].
\end{equation}
Every column of $H$ has weight $2$, every row has weight $q+1$, and the apartment words are the six-edge hexagons in the incidence graph. This example shows geometrically why apartment cycles attain the minimum distance.

For $q=2$, the code parameter is given by $[21, 8, 6]$. Its parity-check matrix and Tanner-graph automorphisms are given in Appendix~\ref{appendix:code-examples}.

\paragraph{Explicit subfamily: $G=\operatorname{SL}_{r+1}(q)$} \label{par:subfamily}
The above example is just an instance of a sequence of Chevalley groups $G=\operatorname{SL}_{r+1}(q)$ with the Lie rank $r$. Its Borel subgroup is again the upper triangular matrices. The Weyl group is the permutation group $S_{r+1}$, and thus $|W|=(r+1)!$. The size $|B|$ of the Borel subgroup is given by $q^{r(r+1)/2}(q-1)^r$, the number of chambers $[G:B]$ is given by~\cite{Morrison2006}
\begin{align}
    [G:B] &=\prod_{i=1}^{r+1}(1+q+\cdots q^{i-1}) = \prod_{i=2}^{r+1} \frac{q^i-1}{q-1},
\end{align}
and the maximum length $\max_{w\in W}l(w)$ of the Weyl group is $r(r+1)/2$ (reverse permutation). Thus, as both $q$ and $r$ grows, the code parameter scales as 
\begin{equation}
[\leq q^{r(r+1)/2}e^{r/q+r/q^2},q^{r(r+1)/2},(r+1)!], \label{eq:asymptotic_scaling}
\end{equation}
where the first term follows from
\begin{align}
    \ln[G:B] &= \sum_{i=2}^{r+1} [\ln q^{i-1} +\ln(1-q^{-i}) - \ln(1-q^{-1})] \\
    &\leq \frac{r(r+1)}{2}\ln q - \left(\sum_{i=2}^{r+1} q^{-i}\right) + r\left(\frac{1}{q}+\sum_{j=2}^{\infty}\frac{q^{-j}}{j}\right) \\ 
    &\leq\frac{r(r+1)}{2}\ln q + \frac{r}{q} + \frac{r}{q^2}. 
\end{align}
The last inequality holds whenever $q\geq 2$.
Thus, by setting $r=cq$ with a constant $c$ and letting $q$ grow, this code family achieves a constant rate. The distance, on the other hand, grows only subpolynomially (yet faster than quasi-polylogarithmically) in $n$, since $d\sim q^{cq}$ compared to $n\sim q^{c^2 q^2}$. The row and column weights also grow, but very slowly. More precisely, it grows as $(\log n)^{1/2-o(1)}$.

\begin{rem} \label{rem:meta-code}
    Although the only information of the spherical Tits building $\Delta$ used for the construction is vector spaces $C_{r-1}(\Delta;\F_2)$, $C_{r-2}(\Delta;\F_2)$, and the homology space $\widetilde{H}_{r-1}(\Delta;\F_2)$, the fact that $\widetilde{H}_i(\Delta;\F_2)=0$ for $i<r-1$, which follows from the Solomon-Tits theorem~\cite{Abramenko2008}, may be useful for the syndrome consistency. For $r\geq 2$, consider the first three chains of the Tits building $\Delta$ as 
    \begin{equation}
        C_{r-1}(\Delta;\F_2) \xrightarrow{\partial} C_{r-2}(\Delta;\F_2) \xrightarrow{\partial'} C_{r-3}(\Delta;\F_2), \label{eq:length-three-chain}
    \end{equation}
    where $C_{-1}(\Delta;\F_2)\coloneqq \F_2$ for convenience.
    Then, since $\ker{\partial'} = \im\partial$, the syndrome space $\im\,\partial$ of the Steinberg code can be regarded as the codewords of another code, $\ker{\partial'}$. This structure can be exploited to construct meta-checks~\cite{Higgott2023, Ostrev2024} when constructing a quantum code.

    Note that the code distance of $\ker{\partial'}$ is given by $r$. This is because the row weight of $\partial$ is $r$, which upper-bounds the distance of $\ker{\partial'}$, and any $(r-2)$-cycle in a simplicial complex has at least $r$ simplices, which lower-bounds the distance. This is typically much smaller than the code distance $|W|$ of the Steinberg code.
\end{rem}

\subsection{Binary symplectic representation of $\overline{\mathrm{Cl}_n/\mathcal{P}_n}$} \label{sec:symplectic-rep}

In this section, we introduce the binary symplectic representation of the $n$-qubit Clifford group~\cite{Dehaene2003, Aaronson2004}, which is a convenient notation used throughout the paper for our construction of full logical Clifford gates.
Let $\GL_n(\F_2)$ denote the set of $n \times n$ invertible binary matrices.

\paragraph{Binary representation of Pauli operators.}
The $n$-qubit Pauli group ${\cal P}_n$ can be written as
\begin{equation}
{\cal P}_n\coloneqq \left\{
i^{\kappa}X^{\bm{x}}Z^{\bm{z}}
\;\middle|\;
\kappa\in\mathbb{Z}_4,\
\bm{x},\bm{z}\in\F_2^n
\right\},
\end{equation}
where
\begin{equation}
X^{\bm{x}}Z^{\bm{z}}\coloneqq\bigotimes_{j=1}^{n}X_j^{x_j}Z_j^{z_j}.
\end{equation}
We identify Pauli operators up to a global phase by taking the projective version $\overline{\mathcal{P}}_n$ of $\mathcal{P}_n$, where two Pauli operators are equivalent if they differ up to a phase.
If one writes 
\begin{equation}        
    \forall \bm{v}=(\bm{x}, \bm{z})\in\F_2^{2n}, \qquad P(\bm{v})=X^{\bm{x}}Z^{\bm{z}},
\end{equation} 
then any element of $\overline{\mathcal{P}}_n$ can be uniquely identified as a representative $[P(\bm{v})]$ with a binary vector $\bm{v}\in\F_2^{2n}$.

Multiplication of Pauli operators corresponds to addition of their binary labels:
\begin{equation}
[P(\bm{v})] [P(\bm{w})] = [P(\bm{v}+\bm{w})], \qquad \bm{v},\bm{w}\in\F_2^{2n}.
\end{equation}
Hence, we identify
\begin{equation}
\overline{\mathcal{P}}_n
\cong
\F_2^{2n}
\end{equation}
as additive groups.

The commutation relation of two Pauli operators 
\begin{equation}
    P(\bm{v})P(\bm{w}) = (-1)^{\langle\bm{v},\bm{w}\rangle_{\Omega}} P(\bm{w})P(\bm{v})
\end{equation}
is encoded by the binary symplectic form defined by
\begin{align}
\langle\bm{v},\bm{w}\rangle_{\Omega}
&=
\bm{v}\Omega\bm{w}^T,
\end{align}
where
\begin{equation}
\Omega = 
\begin{pmatrix}
0&I_n \\
I_n&0
\end{pmatrix}.
\end{equation}
Thus, $P(\bm{v})$ and $P(\bm{w})$ commute iff $\langle\bm{v},\bm{w}\rangle_{\Omega}=0$.

\paragraph{Binary symplectic representation of a Clifford operator.}
The $n$-qubit Clifford group $\mathrm{Cl}_n$ is the normalizer of $\mathcal{P}_n$ in the $n$-qubit unitary group $\mathrm{U}(2^n)$:
\begin{equation}
\left\{
U\in\mathrm{U}(2^n)
\;\middle|\;
U\mathcal{P}_nU^\dagger=\mathcal{P}_n
\right\}.
\end{equation}
Since ${\cal P}_n$ is a normal subgroup of $\mathrm{Cl}_n$, the quotient $\mathrm{Cl}_n/{\cal P}_n$ also forms a group. When two Clifford operators differ only by a Pauli operator, their conjugation actions on Pauli operators differ at most by signs. Consequently, after discarding Pauli phases, the two Clifford operators induce the same transformation. 
Here, we do not care about a global phase, so we consider the projective quotient $\overline{\mathrm{Cl}_n/\mathcal{P}_n}$ of $\mathrm{Cl}_n/\mathcal{P}_n$.

For every $U\in\mathrm{Cl}_n$, conjugation maps each Pauli operator to another Pauli operator. It therefore induces a bijection $f_U:\F_2^{2n}\to\F_2^{2n}$ of the binary label space defined by
\begin{equation}
    [UP(\bm{v})U^{\dagger}]= [P(f_U(\bm{v}))].
\end{equation}
Using multiplication of Pauli classes, one obtains
\begin{align}
[P(f_U(\bm{v}+\bm{w}))]
&=
[UP(\bm{v}+\bm{w})U^\dagger]
\\
&=
[UP(\bm{v})U^\dagger]
[UP(\bm{w})U^\dagger]
\\
&=
[P(f_U(\bm{v})+f_U(\bm{w}))].
\end{align}
Because binary labels are unique, this implies
\begin{equation}
f_U(\bm{v}+\bm{w}) = f_U(\bm{v})+f_U(\bm{w}).
\end{equation}
Thus, $f_U$ is linear over $\F_2$. Consequently, there is an invertible binary matrix $F_U \in \GL_{2n}(\F_2)$ such that
\begin{equation}
f_U(\bm{v})=\bm{v}F_U^T.
\end{equation}

Since conjugation by a unitary preserves commutation and anticommutation relations, the induced map $f_U$ must preserve the symplectic form. Thus, the following holds
\begin{equation}
    \langle \bm{v}F_U^T,\bm{w}F_U^T\rangle_{\Omega} = \langle\bm{v},\bm{w}\rangle_{\Omega}
\end{equation}
for all $\bm{v},\bm{w}\in\F_2^{2n}$. Therefore,
\begin{equation}
F_U^T\Omega F_U=\Omega.
\end{equation}
A binary matrix satisfying this condition is called symplectic. The binary symplectic group $\mathrm{Sp}_{2n}(\F_2)$ is
\begin{equation}
\mathrm{Sp}_{2n}(\F_2) \coloneqq \left\{
F\in\operatorname{GL}_{2n}(\F_2)
\;\middle|\;
F^T\Omega F=\Omega
\right\}.
\end{equation}

\paragraph{Identification of the quotient.}

The assignment
\begin{equation}
\begin{split}
\Phi:
\overline{\mathrm{Cl}}_n
&\longrightarrow
\operatorname{Sp}_{2n}(\F_2),
\\
[U]&\longmapsto F_U,
\end{split}
\end{equation}
is a group homomorphism since $F_{UV}=F_U F_V$ holds.
Every Pauli operator lies in its kernel. Indeed, if $Q=P(\bm{q})\in \overline{\mathrm{Cl}}_n$, then
\begin{equation}
 QP(\bm{v})Q^{\dagger} = (-1)^{\langle\bm{q},\bm{v}\rangle_{\Omega}}
P(\bm{v}),
\end{equation}
so conjugation by $Q$ changes only the phase of $P(\bm{v})$ and leaves its binary label unchanged. In particular, $U$ and $UQ$ determine the same symplectic matrix.

Conversely, suppose that $F_U=I_{2n}$. Then conjugation by $U$ fixes the binary label of every Pauli operator and can only change its sign. These signs form a character of the binary Pauli space. Because the symplectic form is nondegenerate, there is a vector $\bm{q}$ such that these signs are exactly those produced by conjugation with $P(\bm{q})$. The operator $P(\bm{q})^\dagger U$ then commutes with every Pauli operator. Since the Pauli operators span the full matrix algebra, $P(\bm{q})^\dagger U$ must be a scalar multiple of the identity. Thus, projectively, $U$ is a Pauli operator, and
\begin{equation}
\ker\Phi=\overline{\mathcal{P}}_n.
\end{equation}

Finally, every binary symplectic transformation can be realized by a Clifford operator. For example, the Hadamard, phase, and CNOT gates induce elementary symplectic transformations that generate $\operatorname{Sp}_{2n}(\F_2)$ \cite{Dehaene2003}. Hence, $\Phi$ is surjective, giving that
\begin{equation}
\overline{\mathrm{Cl}_n/\mathcal{P}_n} = \overline{\mathrm{Cl}_n} /
\overline{\mathcal{P}_n}
\cong
\operatorname{Sp}_{2n}(\F_2)
.
\end{equation}
The symplectic matrix $F_U$ records the action of $U$ on the binary Pauli labels but does not record the phase function $\epsilon_U(\bm{v})$ appearing in the following equation:
\begin{equation}
    UP(\bm{v})U^{\dagger} = \epsilon_U(\bm{v})P(\bm{v}F_U^T).
\end{equation}
Changing those phases corresponds precisely to multiplying $U$ by a Pauli operator, which explains why $F_U$ represents a coset in $\mathrm{Cl}_n/\mathcal{P}_n$ rather than an individual Clifford operator.

\subsection{Hypergraph product codes} \label{sec:hypergraph_product}
In this section, we summarize how to construct a quantum CSS code from classical codes using the hypergraph product construction~\cite{Tillich2014}.
From two classical linear codes with parity check matrices $H_A \in \F_2^{m_a \times n_a}$ and $H_B \in \F_2^{m_b \times n_b}$, respectively, the hypergraph product $\text{HGP}(H_A,H_B)$ of $H_A$ and $H_B$ is constructed by taking a homological product of chain complexes $\F_2^{n_a} \xrightarrow{H_A} \F_2^{m_a}$ and $\F_2^{m_b} \xrightarrow{H_B^T} \F_2^{n_b}$ to form a chain complex $C_2 \xrightarrow{H_Z^T} C_1 \xrightarrow{H_X} C_0$ of length-2 as follows.
\begin{equation}
\begin{tikzcd}
    & \F_2^{n_a} \otimes \F_2^{n_b} 
    \arrow[dr, "H_A \otimes I_{n_b}"] \\
    \F_2^{n_a} \otimes \F_2^{m_b} 
    \arrow[ur, "I_{n_a} \otimes H_B^T"] 
    \arrow[dr, "H_A \otimes I_{m_b}"']  
    && \F_2^{m_a} \otimes \F_2^{n_b} \\
    & \F_2^{m_a} \otimes \F_2^{m_b}
    \arrow[ur, "I_{m_a} \otimes H_B^T"']
\end{tikzcd}
\end{equation}
The parity check matrices $H_Z$ and $H_X$ of $\text{HGP}(H_A, H_B)$, where each row corresponds to a stabilizer generator, are thus given by
\begin{align}
    H_Z &= (I_{n_a} \otimes H_B ~|~ H_A^T \otimes I_{m_b}), \label{eq:Z_stabilizer_generator}\\  
    H_X &= (H_A \otimes I_{n_b} ~|~ I_{m_a} \otimes H_B^T),\label{eq:X_stabilizer_generator}
\end{align}
which satisfies the CSS condition $H_X H_Z^T=0$. We call the physical qubits corresponding to the columns of the left-hand side of the bar $|$ the left-sector physical qubits and those corresponding to the right-hand side the right-sector physical qubits.

Let $V^*$ be the dual space of $V \subset \F_2^{n}$. 
The set of logical Pauli operators of a HGP code $\text{HGP}(H_A, H_B)$ is then given by
\begin{align}
    \mathcal{L}_Z &= \{(g_A \otimes \bar{g}_B ~|~ 0): g_A^T \in \ker{H_A},~ \bar{g}_B^T \in (\ker{H_B})^* \} \\
    \cup& 
    \{(0 ~|~ \bar{g}'_A \otimes g'_B): \bar{g}'^T_A \in (\ker{H_A^T})^*,~ g'^T_B \in \ker{H_B^T} \}, \\
    \mathcal{L}_X &= \{(\bar{g}_A \otimes g_B ~|~ 0): \bar{g}_A^T \in (\ker{H_A})^*,~ g_B^T \in \ker{H_B} \} \\
    \cup& 
    \{(0 ~|~ g'_A \otimes \bar{g}'_B): g'^T_A \in \ker{H_A^T},~ \bar{g}'^T_B \in (\ker{H_B^T})^* \},
\end{align}
which are also separated into the left and right sectors~\cite{Quintavalle_2023}. 
We can also give the generator matrices of these logical Pauli operators. Let $G_A, G_B$ be the $k_a\times n_a$ and $k_b \times n_b$ generator matrices of the two classical codes, i.e., $H_A G_A^T = 0 = H_B G_B^T$. Then, there exist matrices $G_A^*$ and $G_B^*$ such that the left-sector logical Pauli operators are generated by 
\begin{align}
    G_Z^{(L)} &= (G_A \otimes G_B^* ~|~ 0), \label{eq:lshgp-logical-basis-Z} \\
    G_X^{(L)} &= (G_A^* \otimes G_B ~|~ 0),\label{eq:lshgp-logical-basis-X}  \\
    \intertext{and}
    G_{A(B)} (G_{A(B)}^* )^T &= G_{A(B)}^* G_{A(B)}^T = I_{k_{a(b)}}.
\end{align}
In fact, $G_{A(B)}^*$ is given by the transpose of the right inverse of $G_{A(B)}$. The generators of the right-sector logical operators can also be written in the same way.

The minimum weight of the left-sector logical Pauli operators, i.e., $\{u(G_Z^{(L)})^T:u\in\F_2^{n_a n_b + m_a m_b}\}\cup \{u(G_X^{(L)})^T:u\in\F_2^{n_a n_b + m_a m_b}\}$, is given by $\min\{d_a,d_b\}$, where $d_a$ and $d_b$ are the code distances of the classical code $\ker{H_A}$ and $\ker{H_B}$, respectively~\cite{Quintavalle_2023}. Similarly, the minimum weight of the right-sector logical Pauli operators is given by $\min\{d_a^T,d_b^T\}$, where $d_a^T$ and $d_b^T$ are the minimum distances of $\ker{H_A^T}$ and $\ker{H_B^T}$, respectively.

The code $\text{HGP}(H_A, H_B)$ thus has the code parameter $\llbracket n_a n_b+m_a m_b,k_a k_b + k_a^T k_b^T,\min\{d_a,d_b,d_a^T,d_b^T\}\rrbracket$, where $k_i^T=m_i-(n_i-k_i)$. For our later purposes, however, we only care about the left-sector logical qubits and gauge out the right-sector logical qubits. We call such a code $\text{LS-HGP}(H_A, H_B)$.
\begin{dfn}
    The CSS code $\text{LS-HGP}(H_A, H_B)$ is defined as keeping the left-sector logical qubits of $\text{HGP}(H_A, H_B)$ and gauging its right-sector logical qubits. 
\end{dfn}
Thus, strictly speaking, $\text{LS-HGP}(H_A, H_B)$ is a subsystem code while its gauge degree of freedom is not actively used in the following analysis.
Since the right-sector logical operators are no longer relevant, the code parameter of $\text{LS-HGP}(H_A, H_B)$ is given by $\llbracket n_a n_b+m_a m_b,k_a k_b,\min\{d_a, d_b\}\rrbracket$. 
Indeed, all the transversal and the fold-transversal gates that we use throughout this paper never couple the left- and right-sector logical qubits~\cite{Quintavalle_2023, berthusen2025automorphismgadgetshomologicalproduct}, and thus the minimum weight $\min\{d_a, d_b\}$ of the left-sector logical Pauli operators Eqs.~\eqref{eq:lshgp-logical-basis-Z} and \eqref{eq:lshgp-logical-basis-X} is not reduced by any logical operations induced by them.

\subsubsection{Automorphisms inherited from a classical code} \label{sec:single-block-diagonal}
We review how an automorphism of a Tanner graph of a classical code gives rise to a Clifford automorphism of the corresponding HGP code~\cite{Hong2024, berthusen2025automorphismgadgetshomologicalproduct}. Let $\mathrm{Aut}(H)$ be the group of automorphisms of a Tanner graph corresponding to a parity-check matrix $H\in \F_2^{m \times n}$ of a classical code $\ker{H}$. Let $G \in \F_2^{k \times n}$ be a generator matrix of this code, with $k=n-{\rm rank}(H)$. For a coordinate permutation $\sigma \in \mathrm{Aut}(H)$, let $P^{\sigma}\in\F_2^{n\times n}$ be a permutation matrix defined as
\begin{equation}
    \sigma: v \mapsto v P^{\sigma}. 
\end{equation}
Since $\sigma$ is a Tanner-graph automorphism, there exists another permutation matrix $Q^{\sigma}\in\F_2^{m\times m}$ acting on the rows of $H_A$ as 
\begin{equation}
    H P^{\sigma} = Q^{\sigma} H. \label{eq:coordinate_check_permutation}
\end{equation}
A Tanner-graph automorphism may act nontrivially on codewords, i.e., there exists $\pi(\sigma)\in\mathrm{GL}_k(\F_2)$ such that
\begin{equation}
    G P^{\sigma} = \pi(\sigma)G.
\end{equation}
In this sense, the code space is a representation space of $\mathrm{Aut}(H)$.

We now consider a HGP code $\mathrm{HGP}(H_A, H_B)$ with parity-check matrices given in Eqs.~\eqref{eq:Z_stabilizer_generator} and~\eqref{eq:X_stabilizer_generator}. Then, there are two natural lifts of classical Tanner-graph automorphisms discussed below. Before writing them explicitly, recall that an invertible binary matrix $L\in \GL_{n_a n_b+m_a m_b} (\F_2)$ determines a $Z$- and $X$-preserving Clifford transformation whose binary symplectic representation is
\begin{equation}
    \mathcal{D}(L) = \begin{pmatrix}L & 0 \\ 0 & L^{-T} \end{pmatrix}.
\end{equation}
The inverse transpose in the $Z$ block follows from preserving the Pauli commutation relation. Thus, to obtain a logical Clifford automorphism of the HGP code, it is enough to show that 
\begin{equation}
    \mathrm{row}(H_X L) = \mathrm{row}(H_X), \qquad \mathrm{row}(H_Z L^{-T}) = \mathrm{row}(H_Z), \label{eq:preservation_rule}
\end{equation}
where $\mathrm{row}(H)$ denotes the row space of $H$. These conditions say exactly that this Clifford normalizes the $X$- and $Z$-stabilizer groups.

For a Tanner-graph automorphism $\sigma_a\in{\rm Aut}(H_A)$, define the first-tensor-factor lift by
\begin{equation}
    L_A(\sigma_a)=(L_A(\sigma_a))^{-T} = (P_A^{\sigma_a} \otimes I_{n_b}) \oplus (Q_A^{\sigma_a} \otimes I_{m_b}). \label{eq:left_tensor_factor}
\end{equation}
In the same way, for $\sigma_b\in{\rm Aut}(H_B)$ the second-tensor-factor lift is defined as 
\begin{equation}
    L_B(\sigma_b)=(L_B(\sigma_b))^{-T} = (I_{n_a}\otimes P_B^{\sigma_b}) \oplus (I_{m_a} \otimes Q_B^{\sigma_b}). \label{eq:right_tensor_factor}
\end{equation}
We now verify Eq.~\eqref{eq:preservation_rule}. For the first-tensor-factor lift, we have
\begin{align}
    H_X L_A(\sigma_a) &= (H_A P_A^{\sigma_a} \otimes I_{n_b}~|~Q_{A}^{\sigma_a}\otimes H_B^T) \\
    &= (Q_A^{\sigma_a}H_A  \otimes I_{n_b}~|~Q_{A}^{\sigma_a}\otimes H_B^T) \\
    &= (Q_A^{\sigma_a}\otimes I_{m_b})H_X.
\end{align}
This shows $\mathrm{row}(H_X L_A(\sigma_a))=\mathrm{row}(H_X)$. For the $Z$ checks, we have
\begin{align}
    H_Z (L_A(\sigma_a))^{-T} &= (P_A^{\sigma_a}\otimes H_B~|~H_A^T Q_A^{\sigma_a} \otimes I_{m_b}) \\
    &= (P_A^{\sigma_a}\otimes H_B~|~P_A^{\sigma_a} H_A^T  \otimes I_{m_b}) \\
    &= (P_A^{\sigma_a}\otimes I_{m_b}) H_Z,
\end{align}
and thus $\mathrm{row}(H_Z (L_A(\sigma_a))^{-T})=\mathrm{row}(H_Z)$. Thus, the first-tensor-factor lift preserves both the $X$ and $Z$ stabilizer subgroups.

The second-tensor-factor lift is completely analogous. Hence, $L_A(\sigma_a)$ and $L_B(\sigma_b)$ both normalize the stabilizer group and therefore define the logical Clifford automorphisms of the HGP code. If we align the physical qubits of $\mathrm{HGP}(H_A, H_B)$ in the square grid, placing the $(i,j)$-th qubit on the $i$-th row and $j$-th column of the grid, then $U_A(\sigma_a)$ only permutes the rows, whereas $U_B(\sigma_b)$ only permutes the columns. In particular, these two actions commute. Thus, $\mathrm{HGP}(H_A, H_B)$ has the two commuting lifts of classical Tanner-graph automorphisms as quantum code automorphisms.

The induced actions of these automorphisms can be given explicitly for the left-sector logical qubits. We have 
\begin{align}
    G_Z^{(L)} L_{A}(\sigma_a) &= (\pi_A(\sigma_a)\otimes I_{k_b}) G_Z^{(L)},\\
    G_Z^{(L)} L_{B}(\sigma_b) &= (I_{k_a} \otimes (\pi_B(\sigma_b))^{-T})G_Z^{(L)},\\
    G_X^{(L)} L_{A}(\sigma_a) &= ((\pi_A(\sigma_a))^{-T}\otimes I_{k_b}) G_X^{(L)},\\
    G_X^{(L)} L_{B}(\sigma_b) &= (I_{k_a} \otimes \pi_B(\sigma_b))G_X^{(L)},
\end{align}
where we used $G^* P^\sigma=(\pi(\sigma))^{-T} G^*$ for the generator $G$ with $G P^\sigma = \pi(\sigma) G$. Thus, its symplectic-matrix representation on the left-sector logical qubits can be given by
\begin{equation}
    \begin{pmatrix}
        (\pi_A(\sigma_a))^{-T}\otimes I_{k_b} & 0 \\ 0 & \pi_A(\sigma_a)\otimes I_{k_b}
    \end{pmatrix},
\end{equation}
for $L_A(\sigma_a)$ and
\begin{equation}
    \begin{pmatrix}
        I_{k_a}\otimes \pi_B(\sigma_b)  & 0 \\ 0 & I_{k_a}\otimes (\pi_B(\sigma_b))^{-T}
    \end{pmatrix},
\end{equation}
for $L_B(\sigma_b)$.

We finally comment on why we stick to Tanner-graph automorphisms of the base classical codes instead of more general code automorphisms. For a classical code with a parity-check matrix $H$, a general code automorphism $\mathrm{Aut}(\ker{H})$ only needs to preserve $\ker{H}$ under a coordinate permutation. This means that, for $\sigma\in\mathrm{Aut}(\ker{H})$, $Q^{\sigma}$ in Eq.~\eqref{eq:coordinate_check_permutation} is not restricted to a permutation matrix but an element of $\mathrm{GL}_m(\F_2)$. Then, since $L_{A(B)}(\sigma_{a(b)})$ has $Q_{A(B)}^{\sigma_{a(b)}}$ action on physical qubits from Eq.~\eqref{eq:coordinate_check_permutation}, which is no longer a physical qubit permutation, $L_{A(B)}(\sigma_{a(b)})$ is not a code automorphism of $\mathrm{HGP}(H_A, H_B)$ in the conventional sense. This is why Ref.~\cite{Malcolm2026} used the subsystem hypergraph product code~\cite{bacon2006, li2020} rather than the usual hypergraph product code; to preserve all the code automorphisms $\mathrm{Aut}(\ker{H})$ of the simplex codes, they needed to use the subsystem version of the hypergraph product in which $H^T$ does not appear in the stabilizer generator. As a cost, the resulting CSS code is not qLDPC as a stabilizer code, although the gauge checks have constant weights.
Here, in contrast, since the automorphism group we want to lift from classical codes is $\mathrm{Aut}(H)$, i.e., Tanner-graph automorphisms, we can use the conventional hypergraph product code.

\subsubsection{$ZX$-duality and fold-transversal gates for $\text{HGP}(H,H)$} \label{subsec:ZX-duality}
Hereafter, we consider the case $H_A=H=H_B$ for the seed classical code of $\text{HGP}(H_A, H_B)$. 
In this case, the parity-check matrices are given by
\begin{align}
    H_Z &= (I_{n} \otimes H ~|~ H^T \otimes I_{m}), \\  
    H_X &= (H \otimes I_{n} ~|~ I_{m} \otimes H^T).
\end{align}
The code $\text{HGP}(H,H)$ has an additional symmetry called the $ZX$-duality~\cite{Breuckmann2024}. Let us index a left-sector physical qubit as $(i,j)\in[n]\times[n]$, corresponding to the $i$-th column of the first tensor factor and the $j$-th column of the second tensor factor, where $[n]=\{1,\ldots,n\}$. In the same way, the right-sector physical qubits can be indexed by $(a,b)\in[m]\times[m]$. 
Throughout this section, $i$ and $j$ denote the indices of the left-sector qubits while $a$ and $b$ denote those of the right-sector qubits.

Now, we define swap maps $\tau_{L}$ and $\tau_{R}$ as 
\begin{align}
    \tau_L: (i,j) &\mapsto (j,i), \\
    \tau_R: (a, b) &\mapsto (b, a),
\end{align}
which correspond to physical-qubit swaps. 
Similar to the physical-qubit index, the $X$ and $Z$ checks are indexed by their row numbers in $H_X$ and $H_Z$, which are in turn identified by the coordinates of the rows of the first and second tensor factors as $(i', a')\in[n]\times [m]$ and $(b', j')\in[m]\times [n]$, respectively. Let $\tau_C$ be defined as 
\begin{equation}
    \tau_C: (i', a') \mapsto (a', i').
\end{equation}
Then, we have
\begin{equation}
    H_X = \tau_C H_Z \tau, \label{eq:exchange_XZ_stabilizer}
\end{equation}
where 
\begin{equation}
    \tau \coloneqq (\tau_L \oplus \tau_R),
\end{equation}
and we abuse the notation by using the same symbol $\tau$ for the map and the matrix.
Indeed, the coordinate exchange gives
\begin{equation}
    \tau_C (I_n\otimes H)\tau_L = H \otimes I_n,
\end{equation}
and 
\begin{equation}
    \tau_C (H^T\otimes I_m)\tau_R = I_m \otimes H^T.
\end{equation}
Thus, $\tau$ maps the support of every $Z$-type stabilizer to the support of a corresponding $X$-type stabilizer. This symmetry is called a $ZX$-duality.

Let us write the swap gates that implement the $\tau$ action in the above as ${\rm SWAP}_{\tau}$. If one places the physical qubit indexed by $(i, j)$ on the $(i, j)$-th entry and $(a, b)$ on the $(a, b)$-th entry of the square grid, where the left- and right-sector physical qubits are placed on different square grids sharing the same diagonal line, then ${\rm SWAP}_{\tau}$ swaps the qubits across the diagonal line of the square grids.

Although the physical qubit swaps ${\rm SWAP}_{\tau}$ alone do not exchange the logical Pauli $Z$ and $X$, combined with the physical transversal Hadamard gates ${\rm H}^{\otimes (n^2+m^2)}$, it realizes a logical Pauli exchange $Z\leftrightarrow X$.  This type of gate $\foldH$ is called the fold-transversal Hadamard gate~\cite{Breuckmann2024}, which is defined as
\begin{equation}
    \foldH\coloneqq {\rm SWAP}_{\tau}\,{\rm H}^{\otimes (n^2+m^2)}.
\end{equation}
This gate $\foldH$ satisfies
\begin{equation}
    \foldH\; Z(v) \; \foldH^{\dagger} = X(v\tau),
\end{equation}
where $v\in\F_2^{n^2+m^2}$, and $Z(v)$ denotes a multi-qubit Pauli-$Z$ operator that acts on the qubits corresponding to the elements $1$ in $v$. Here, we suppose that $v$ is indexed by the left-sector index $(i,j)$ and the right-sector index $(a,b)$ with the correspondence given above.
Thus, the $Z$- and $X$-type stabilizer subgroup ${\cal S}_Z$ and ${\cal S}_X$ satisfy
\begin{equation}
    \foldH\; {\cal S}_Z\; \foldH^{\dagger} = {\cal S}_X, \qquad \foldH\; {\cal S}_X\; \foldH^{\dagger} = {\cal S}_Z,
\end{equation}
and the logical operators ${\cal L}_Z$ and ${\cal L}_X$ satisfy
\begin{equation}
    \foldH\; {\cal L}_Z\; \foldH^{\dagger} = {\cal L}_X, \qquad \foldH\; {\cal L}_X\; \foldH^{\dagger} = {\cal L}_Z.
\end{equation}
Under the action of $\foldH$, the logical qubits are Hadamard transformed and swapped according to $\tau_{k^2}:[k]\times [k]\ni (\alpha, \beta)\mapsto(\beta,\alpha)$ for the left sector and $\tau_{(k^T)^2}:[k^T]\times [k^T]\ni (\alpha', \beta')\mapsto(\beta',\alpha')$ for the right sector. This can be checked from
\begin{equation}
    \tau_{k^2} G_Z^{(L)} \tau = G_X^{(L)}, 
\end{equation}
and similarly for the right-sector generator matrices.
Its symplectic-matrix representation on the left-sector logical qubits is thus given by
\begin{equation}
    \begin{pmatrix}
        0 & \tau_{k^2} \\ \tau_{k^2} & 0
    \end{pmatrix}.
\end{equation}

There is another fold-transversal gate exploiting this $ZX$-duality, the phase-type fold-transversal gate $\foldS$, which is defined as~\cite{Quintavalle_2023}
\begin{equation}
    \foldS\coloneqq \prod_{i=1}^n \operatorname{S}_{(i,i)} \prod_{i<j} {\rm C}Z_{(i,j),(j,i)} \otimes \prod_{a=1}^m \operatorname{S}^{\dagger}_{(a,a)} \prod_{a<b} {\rm C}Z_{(a,b),(b,a)}.
    \label{eq:fold-S}
\end{equation}
It satisfies 
\begin{align}
    \foldS\; Z(u) \; \foldS^{\dagger} &= Z(u), \\
    \foldS\; X(v) \; \foldS^{\dagger} &= i^{f_{\operatorname{S}}(v)}X(v)Z\bigl(v\tau\bigr),
\end{align}
for any $u,v\in\F_2^{n^2+m^2}$, where $f_{\operatorname{S}}(v)$ is defined as 
\begin{equation}
    f_{\operatorname{S}}(v)\coloneqq \sum_{i=1}^{n} v_{(i,i)} - \sum_{a=1}^m v_{(a,a)},
\end{equation}
which is an integer sum, not a binary sum.
The choice of $\operatorname{S}$ and $\operatorname{S}^{\dagger}$ in Eq.~\eqref{eq:fold-S} is to ensure~\cite{Quintavalle_2023} 
\begin{equation}
    \forall S_x \in {\cal S}_X, \qquad \foldS\;S_x\;\foldS^{\dagger}=S_x S_{z(x)}, \quad \exists S_{z(x)} \in {\cal S}_Z.
\end{equation}
From this, one can see that $\foldS$ preserves the stabilizer group $\langle{\cal S}_Z,{\cal S}_X\rangle$. On logical qubits, $\foldS$ acts as ${\rm C}Z$ between $(\alpha,\beta)\in[k]\times[k]$ and $(\beta,\alpha)\in[k]\times[k]$ while acts as $\operatorname{S}$ on $(\alpha,\alpha)\in[k]\times[k]$, both up to phases. Its symplectic-matrix representation on the left-sector logical qubits is thus given by 
\begin{equation}
    \begin{pmatrix}
        I_{k^2} & \tau_{k^2} \\ 0 & I_{k^2}
    \end{pmatrix}.
\end{equation}

\section{Code construction with all the logical Clifford gates} 
\label{sec:construction}
\subsection{Construction}
We first describe a family of HGP codes with code automorphisms lifted from the seed classical codes and then show that they can perform all in-block and inter-block logical Clifford gates with transversal and fold-transversal physical Clifford gates. Although this family is not strictly qLDPC, the codes have relatively low stabilizer weight $\lesssim 10$ for a practical number $\lesssim 10^{18}$ of physical qubits and a distance that grows much faster than the stabilizer weight by adjusting the parameters. Thus, although the code family considered here is not an asymptotic qLDPC family, it simultaneously achieves a high rate and a low stabilizer weight in practical sizes.

The code family considered here is $\text{LS-HGP}(H_{\rm St}, H_{\rm St})$, where $H_{\rm St}$ is a parity-check matrix of a Steinberg code introduced in Sec.~\ref{sec:steinberg_code}. We name this family of codes \nameofcode codes after ``PRoduct codes from Irreducible Steinberg Modules,'' as stated in the introduction. Since it only leaves the left-sector logical qubits, the code parameters of $\text{LS-HGP}(H_{\rm St}, H_{\rm St})$ can be explicitly given by $\llbracket n^2+m^2, k^2, d\rrbracket$ when $H_{\rm St}$ is an $ m \ times n$ matrix for an $[n, k, d]$ Steinberg code. The stabilizer weight of this code is bounded from above by $w_{\mathrm{col}}+w_{\mathrm{row}}$ of $H_{\rm St}$, and the number of stabilizers that act on a particular physical qubit is bounded from above by $\max\{w_{\mathrm{col}},w_{\mathrm{row}}\}$ of $H_{\rm St}$. 
If one uses a subfamily $\operatorname{SL}_{r+1}(q)$, given in the example~\ref{par:subfamily}, one gets a subfamily of the \nameofcode codes with the code parameter 
\begin{equation}
    \left\llbracket \left(1+\frac{r^2}{(q+1)^2}\right)\!  \left(\prod_{i=2}^{r+1} \frac{q^i-1}{q-1}\right)^2\!,q^{r(r+1)}, (r+1)!\right\rrbracket,
\end{equation}
and the stabilizer weight $r+q+1$.
Thus, by scaling up $q$ and $r$ while keeping $r\simeq q$ and using the asymptotic expansion Eq.~\eqref{eq:asymptotic_scaling}, one obtains the following scaling of code parameters:
\begin{equation}
    \llbracket n, \Theta(n), n^{o(1)} \rrbracket,
\end{equation}
with the stabilizer weight $(\log n)^{1/2-o(1)}$.
More explicitly, by setting $r=cq$ for a constant $c$, the encoding rate approached $e^{-2c}/(1+c^2)$ from Eq.~\eqref{eq:asymptotic_scaling}.
Thus, it is a constant-rate subfamily.
The distance, on the other hand, grows only subpolynomially, more precisely, $\exp[(\log n)^{1/2+o(1)}]$. 
Small instances for this $\operatorname{SL}_{r+1}(q)$ subfamily includes $\llbracket 637, 64, 6\rrbracket$ code ($q=r=2$), $\llbracket 12789, 4096, 6\rrbracket$ code ($q=4, r=2$), and $\llbracket 198450, 4096, 24\rrbracket$ code ($q=2,r=3$).

As we mentioned, this code family is not strictly qLDPC. However, the factorial scaling of the number of physical qubits is universally true for the \nameofcode codes. Thus, the numbers of physical and logical qubits blow up very quickly compared to the stabilizer weight. Thus, in practical sizes, \nameofcode codes can be treated as almost-qLDPC codes. For example, if the seed classical code is chosen to be $\operatorname{SL}_{r+1}(q)$ as given above, then setting $r=q+1=5$ already leads to the number of physical qubits $\approx 10^{18}$ while the stabilizer weights are $\leq 10$.

\begin{rem}\label{rem:meta-check}
    The \nameofcode code has a meta-check, inherited from a syndrome consistency check for the Steinberg code as given in Remark~\ref{rem:meta-code}. In fact, by writing a corresponding parity-check matrix for $\partial'$ in Eq.~\eqref{eq:length-three-chain} $M_{\rm St}$, we can see that $M_{\rm St}\otimes G_{\rm St}$ and $G_{\rm St}\otimes M_{\rm St}$ serve as meta-checks for the $X$ and $Z$ syndromes of $\text{LS-HGP}(H_{\rm St}, H_{\rm St})$, respectively, since $M_{\rm St} H_{\rm St}=0$ and $G_{\rm St}H_{\rm St}^T=0$. The code distance of this meta-check is given by $\min\{r,q+1\}$, where the latter comes from the code distance of $\ker{G_{\rm St}}$. Indeed, since any support of at most $q$ chambers has a common opposite~\cite{Bjorner1984}, where the opposite chamber $hB$ for a chamber $gB$ means $g^{-1}h\in Bw_0B$ with a longest Weyl element $w_0$, an apartment chain can intersect that support in exactly one chamber. However, $\im H_{\rm St}^{T}$ is orthogonal to any apartment chain, so any vector in $\im H_{\rm St}^{T}$ has at least $q+1$ nonzero elements.

    This meta-check structure may help its error-correction routine, especially for the encoded-state preparation.
\end{rem}

\subsection{All the logical Clifford gates from transversal and fold-transversal gats} \label{subsec:Clifford-synthesis}

As explained in Sec.~\ref{sec:hypergraph_product}, a \nameofcode code defined in the previous section inherits code automorphisms, which act as a tensor product of absolutely irreducible Steinberg representations on the logical Pauli operators. It also has a $ZX$-duality given in Sec.~\ref{sec:hypergraph_product}. Furthermore, since a \nameofcode code is CSS, the physical transversal CNOT gates between two \nameofcode code blocks induce logical parallel CNOT gate actions. In this section, we show that this set of elementary gate gadgets is in fact sufficient to generate all addressable logical Clifford gates on code blocks greater than one, when combined with logical Pauli actions that can trivially be performed transversally. As explained in Sec.~\ref{sec:symplectic-rep}, this amounts to showing that the symplectic representations of these gate primitives generate all the binary symplectic matrices.

We first introduce the notation for multi-block symplectic representation of Clifford gates. Consider $b(>1)$ blocks of $\llbracket n',k',d\rrbracket$ \nameofcode codes.
Let $k'\geq 3$, which is always satisfied for the \nameofcode codes, and put $m=bk'$.  We order the symplectic coordinates as
\begin{equation}
 (\bm{x}_1,\ldots,\bm{x}_{b},\bm{z}_1,\ldots,\bm{z}_{b}),
 \qquad \bm{x}_i,\bm{z}_i\in\mathbb F_2^{k'},
\end{equation}
and use
\begin{align}
 \Omega_m&=\begin{pmatrix}0&I_m\\ I_m&0\end{pmatrix},\\
 \Sp_{2m}(\F_2)&=\{S\in\GL_{2m}(\F_2):S^{T}\Omega_mS=\Omega_m\}.
\end{align}
Let $\operatorname{Sym}_m(\F_2) \subset M_m(\F_2)$ be the set of symmetric matrices. For $A\in\GL_m(\F_2)$ and $B,C\in \operatorname{Sym}_m(\F_2)$, define
\begin{align}
 \mathcal D(A)&=
 \begin{pmatrix}A&0\\0&A^{-T}\end{pmatrix}, \\
 \mathcal U(B)&=
 \begin{pmatrix}I_m&B\\0&I_m\end{pmatrix},\\
 \mathcal L(C)&=
 \begin{pmatrix}I_m&0\\C&I_m\end{pmatrix}.
 \label{eq:global-DUL}
\end{align}
These are all symplectic. 

The local (in-block) logical gates of the \nameofcode code can be characterized by invertible matrices $\Gamma\subseteq\GL_{k'}(\F_2)$ and a symmetric, invertible, nonalternating (i.e., having at least one nonzero element on the diagonal) matrix $\tau_{k'}\in M_{k'}(\F_2)$, where $\Gamma$ is given by
\begin{equation}
    \Gamma \coloneqq \{(\pi_{\rm St}(\sigma_1))^{-T} \otimes \pi_{\rm St}(\sigma_2): \sigma_1,\sigma_2\in\mathrm{Aut}(H_{\rm St})\}, \label{eq:def_Gamma}
\end{equation}
and $\tau_{k'}$ corresponds to the tensor-factor swap given in Sec.~\ref{subsec:ZX-duality}.  A gate gadget implemented by a code automorphism induces a logical action given by
\begin{equation}
    \begin{pmatrix}
        G & 0 \\ 0 & G^{-T}
    \end{pmatrix} \in \Sp_{2k'}(\F_2), \qquad G\in\Gamma.
\end{equation}
On the other hand, gate gadgets implemented by $\foldS$ and $\foldH$ are given, respectively, by
\begin{equation}
    \begin{pmatrix}
        I_{k'} & \tau_{k'} \\ 0 & I_{k'}
    \end{pmatrix} \quad \text{and} \quad 
    \begin{pmatrix}
        0 & \tau_{k'} \\ \tau_{k'} & 0
    \end{pmatrix}.
\end{equation}
In a multi-block scenario, we need to specify a block index where a local gate is applied. We thus use the notation $\mathcal D_i(G)$, $\operatorname{S}_{\tau_{k'}, i}$, and $\operatorname{H}_{\tau_{k'}, i}$ for this purpose, implying that it trivially acts on the blocks except the $i$-th block.

From the absolute irreducibility of the Steinberg representation $\pi_{\rm St}$, we have
\begin{equation}
 \operatorname{span}_{\mathbb F_2}\Gamma=M_{k'}(\F_2),
 \label{eq:gamma-span}
\end{equation}
which implies that for any $M\in M_{k'}(\F_2)$,
\begin{equation}
    M = \sum_{s=1}^{\rho_{\Gamma}(M)}h_s, \qquad h_s\in \Gamma, \qquad \rho_{\Gamma(M)}\leq {k'}^2.
\end{equation}
Let $\rho_{\max}$ be defined as
\begin{equation}
    \rho_{\max} \coloneqq \max_M \rho_{\Gamma}(M) \leq k'^2. \label{eq:def_rho_max}
\end{equation}

Since the \nameofcode codes are CSS codes, the physical transversal CNOT gates from the second to the first code block induce logical CNOT action, whose symplectic matrix is given by
\begin{align}
    \mathcal{C}_{12}
    &= \begin{pmatrix}
        I_{k'} &   I_{k'} &  &  & \\
        &  I_{k'} &  &  & \\
        &  &  I_{k'} &  & \\
        &  &  I_{k'} &  I_{k'} 
    \end{pmatrix} \\
    &= \mathcal{D}(I_{2{k'}} + E_{12} \otimes I_{k'}),
\end{align}
where $E_{12}$ is, specifically,
\begin{align}
    E_{12} = \begin{pmatrix}
        0 & 1 \\ 0 & 0
    \end{pmatrix}.
\end{align}
A transversal logical CNOT from the first block to the second block $\mathcal{C}_{21}$ can be constructed similarly by replacing $E_{12}$ with $E_{21}$. In a more general setting, applying transversal CNOT gates from the $j$-th block to the $i$-th block ($i\neq j$) gives 
\begin{align}
    \mathcal{C}_{ij} = \mathcal{D}(I_{m} + E_{ij} \otimes I_{k'}),
\end{align}
where $E_{ij}\in M_{b}(\F_2)$ is defined as 
\begin{equation}
    (E_{ij})_{kl} = \delta_{ik}\delta_{jl}.
\end{equation}

In summary, the symplectic matrix representation of the primitive logical gates of the \nameofcode code is summarized as follows:
\begin{itemize}
    \item The local diagonal gates $\mathcal D_i(G)$ with $G\in \Gamma$, which can be implemented by the code automorphisms of the $i$-th block.
    \item The local shear $\operatorname{S}_{\tau_{k'}, i}$ and Hadamard $\operatorname{H}_{\tau_{k'}, i}$, which can be implemented by the local phase-type and Hadamard-type fold-transversal gates.
    \item The two-block gate $\mathcal C_{ij} = \mathcal D(I_m+E_{ij}\otimes I_{k'})$ with $i\neq j$.
\end{itemize}

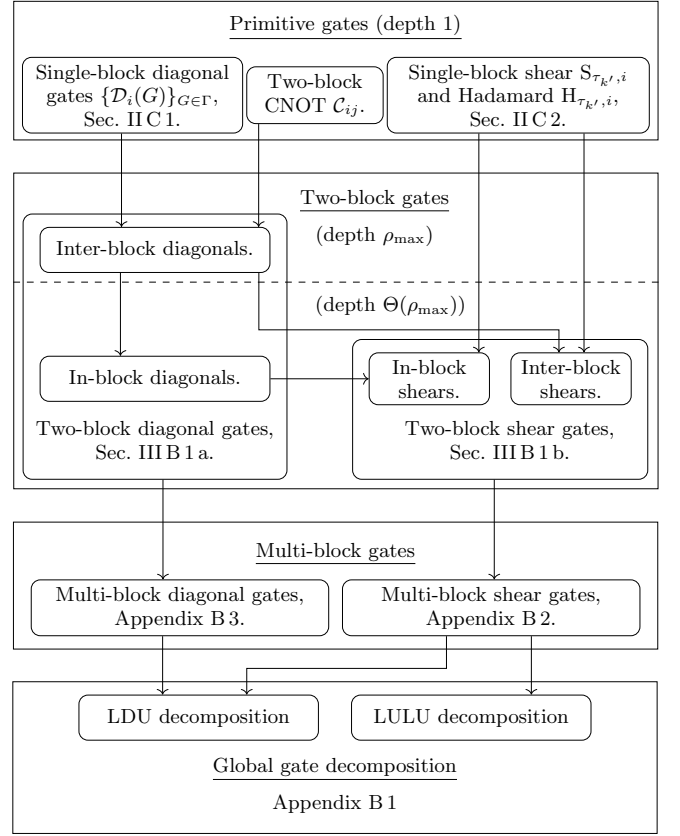
\begin{figure}[!htbp]
    \centering
    \resizebox{\linewidth}{!}{%
    \begin{tikzpicture}[
        flowchart box/.style={
            rectangle,
            rounded corners,
            minimum width=1cm,
            minimum height=0.7cm,
            align=center,
            draw=black
        }
    ]
        
        \node [flowchart box, text width=3.2cm, anchor=west]
        (local-diagonal)
        {\mbox{Single-block diagonal}\\%
        gates $\{\mathcal{D}_i(G)\}_{G\in\Gamma}$,\\%
        Sec.~\ref{sec:single-block-diagonal}.};

        \node [flowchart box, right=0.05cm of local-diagonal, text width=1.9cm]
        (cnot)
        {Two-block\\%
        CNOT $\mathcal{C}_{ij}$.};

        \node [flowchart box, right=0.05cm of cnot, text width=3.85cm]
        (fold)
        {Single-block shear $\mathrm{S}_{\tau_{k'},i}$\\%
        and Hadamard $\mathrm{H}_{\tau_{k'},i}$,\\%
        Sec.~\ref{subsec:ZX-duality}.};

        \node (prim-title) [align=center]
        at ($(local-diagonal.north)!0.3!(fold.north)!0.3!(cnot.north)+(1.2cm,5mm)$)
        {\underline{Primitive gates (depth 1)}};

        \node [draw, fit=(local-diagonal)(cnot)(fold)(prim-title),
        inner sep=1mm, minimum width=10cm] (prim) {};

        \node [draw, below=0.5cm of prim, minimum width=10cm,
        minimum height=4.9cm] (two-block) {};

        \node [align=center, anchor=north]
        (rho-title) at ($(two-block.north)+(0.6cm,-0.15cm)$)
        {\underline{Two-block gates} \\[2mm]%
        (depth $\rho_{\max}$)};

        \node [align=center]
        (two-diagonal-title)
        at ($(two-block.north west)+(2.2cm,-4.15cm)$)
        {Two-block diagonal gates,\\%
        Sec.~\ref{par:two-block-diagonal}.};

        \node [flowchart box, text width=3.35cm]
        (inter-diagonal)
        at ($(two-block.north west)+(2.2cm,-1.2cm)$)
        {Inter-block diagonals.};

        \node [flowchart box, text width=3.35cm]
        (in-diagonal)
        at ($(two-block.north west)+(2.2cm,-3.2cm)$)
        {In-block diagonals.};

        \node [draw, rounded corners,
        fit=(two-diagonal-title)(inter-diagonal)(in-diagonal),
        inner xsep=0.8mm, inner ysep=2mm]
        (two-diagonal) {};

        \node [align=center, anchor=north]
        (two-title) at ($(two-block.north)+(0.85cm,-1.80cm)$)
        {(depth $\Theta(\rho_{\max})$)};

        \node [align=center]
        (two-shear-title)
        at ($(two-block.north)+(2.7cm,-4.15cm)$)
        {Two-block shear gates,\\%
        Sec.~\ref{par:two-block-shear}.};

        \node [flowchart box, text width=1.65cm]
        (in-shear)
        at ($(two-block.north)+(1.45cm,-3.20cm)$)
        {In-block shears.};

        \node [flowchart box, text width=1.65cm]
        (inter-shear)
        at ($(two-block.north)+(3.65cm,-3.20cm)$)
        {Inter-block shears.};

        \node [draw, rounded corners,
        fit=(two-shear-title)(inter-shear)(in-shear),
        inner xsep=2.5mm, inner ysep=2mm]
        (two-shear) {};

        \draw [dashed]
        ($(two-block.north west)+(0,-1.7cm)$) --
        ($(two-block.north east)+(0,-1.7cm)$);

        \node [flowchart box, below=1.4cm of two-block, xshift=-2.42cm,
        text width=4.4cm]
        (multi-diagonal)
        {\mbox{Multi-block diagonal gates,}\\%
        Appendix~\ref{subsec:levi-compilation}.};

        \node [flowchart box, right=0.2cm of multi-diagonal, text width=4.4cm]
        (multi-shear)
        {\mbox{Multi-block shear gates,}\\%
        Appendix~\ref{subsec:shear-synthesis}.};

        \node (multi-title)
        at ($(multi-shear.north)!0.5!(multi-diagonal.north)+(0,4mm)$)
        {\underline{Multi-block gates}};

        \node [draw, fit=(multi-shear)(multi-diagonal)(multi-title),
        inner xsep=1mm, inner ysep=2mm, minimum width=10cm]
        (multi-block) {};

        \node [flowchart box, below=0.7cm of multi-block, xshift=-2.125cm,
        text width=3.5cm]
        (ldu) {LDU decomposition};

        \node [flowchart box, right=0.5cm of ldu, text width=3.5cm]
        (lulu) {LULU decomposition};

        \node (title) [align=center]
        at ($(lulu.south)!0.5!(ldu.south)+(0,-7mm)$)
        {\underline{Global gate decomposition}\\[2mm]%
        Appendix~\ref{subsec:global-gate-decom}};

        \node [draw, fit=(lulu)(ldu)(title),
        inner xsep=1mm, inner ysep=2mm, minimum width=10cm]
        (global) {};

        \coordinate (local-out) at
        ($(local-diagonal.south west)!0.45!(local-diagonal.south east)$);
        \draw [->]
        (local-out) -- (inter-diagonal.north -| local-out);

        \coordinate (cnot-out) at
        ($(cnot.south west)!0.09!(cnot.south east)$);
        \draw [->]
        (cnot-out) -- (inter-diagonal.north -| cnot-out);

        \coordinate (inter-shear-target) at
        ($(inter-shear.north west)!0.4!(inter-shear.north east)$);
        \draw [->]
        ($(inter-diagonal.south west)!0.95!(inter-diagonal.south east)$)
        |- ($(inter-shear-target)+(0,0.35cm)$)
        -- (inter-shear-target);

        \coordinate (fold-inter-out) at
        ($(fold.south west)!0.75!(fold.south east)$);
        \draw [->]
        (fold-inter-out) -- (inter-shear.north -| fold-inter-out);

        \draw [->]
        ($(inter-diagonal.south west)!0.35!(inter-diagonal.south east)$)
        -- ($(in-diagonal.north west)!0.35!(in-diagonal.north east)$);

        \draw [->]
        (in-diagonal.east) -- (in-shear.west);

        \coordinate (fold-in-out) at
        ($(fold.south west)!0.35!(fold.south east)$);
        \draw [->]
        (fold-in-out) -- (in-shear.north -| fold-in-out);

        \coordinate (two-diagonal-out) at
        ($(two-diagonal.south west)!0.53!(two-diagonal.south east)$);
        \draw [->]
        (two-diagonal-out) -- (multi-diagonal.north -| two-diagonal-out);

        \coordinate (two-shear-out) at
        ($(two-shear.south west)!0.48!(two-shear.south east)$);
        \draw [->]
        (two-shear-out) -- (multi-shear.north -| two-shear-out);

        \coordinate (ldu-from-diagonal) at
        ($(ldu.north west)!0.35!(ldu.north east)$);
        \draw [->]
        (multi-diagonal.south -| ldu-from-diagonal)
        -- (ldu-from-diagonal);

        \coordinate (ldu-from-shear) at
        ($(ldu.north west)!0.70!(ldu.north east)$);
        \draw [->]
        ($(multi-shear.south west)!0.35!(multi-shear.south east)$)
        -- ++(0,-0.5cm)
        -| (ldu-from-shear);

        \coordinate (lulu-from-shear) at
        ($(lulu.north west)!0.75!(lulu.north east)$);
        \draw [->]
        (multi-shear.south -| lulu-from-shear)
        -- (lulu-from-shear);
    \end{tikzpicture}
    }
    \caption{A flowchart of the bottom-up construction of full logical Clifford gates for
    \ifcsname nameofcode\endcsname\nameofcode\else the specified\fi\ codes.}
    \label{fig:gate-construction-flowchart}
\end{figure}

\subsubsection{Two-block Clifford gates}

\paragraph{Two-block diagonal gates} \label{par:two-block-diagonal}
In the following, we first consider the two-block case $m=2k'$ and show that \nameofcode codes have all the two-block logical Clifford gates via transversal and fold-transversal gates. We first show that diagonal gates in $\mathrm{Sp}_{4k'}(\mathbb{F}_2)$ can be performed by composing single block-diagonal gates $\mathcal D_{1(2)}(A)$ and two-block transversal CNOT gates.

\begin{prop} \label{prop:diagonal_generation}
    For $k'\geq 3$, the entire set of invertible matrices $\GL_{2k'}(\mathbb{F}_2)$ can be generated by $\mathcal{C}_{12}, \mathcal{C}_{21}$ and single block-diagonal gates $\mathcal{D}_1(A), \mathcal{D}_2(B)$, where $A, B \in \Gamma$. The worst-case gate count is bounded from above by $6(2\rho_{\max} + 1)$, where $\rho_{\max}$ is defined in Eq.~\eqref{eq:def_rho_max}. 
    Therefore, \nameofcode codes realize all diagonal gates in $\mathrm{Sp}_{4k'}(\mathbb{F}_2)$.
\end{prop}
\begin{proof} 
    The first part of the proof follows from Theorem~3, Lemma~2, Theorem~4 of Ref.~\cite{Grassl2013}, combined with Eq.~\eqref{eq:gamma-span} and the fact that $\GL_{2k'}(\F_2)=\mathrm{SL}_{2k'}(\F_2)$. For completeness and to prove the overhead count, we here describe the explicit construction.

    Since we only consider the diagonal gates, we focus on the $2 \times 2$ upper-left block of the overall $4 \times 4$ block matrix.
    For $A\in \Gamma$, its inverse $A^{-1}$ is also in the group $\Gamma$. Direct calculation shows that the following decomposition gives the upper triangular matrix $U(A)$ with $A$ as the right upper block matrix. 
    \begin{align}
        \begin{pmatrix}
            I & A \\ & I
        \end{pmatrix} 
        &= \begin{pmatrix}
            A & \\ & I
        \end{pmatrix} \begin{pmatrix}
            I & I \\ & I
        \end{pmatrix} \begin{pmatrix}
            A^{-1} & \\ & I
        \end{pmatrix}.
        \label{eq:upper-triangular-decom}
    \end{align}
    Note that $U(A)$ is different from a shear gate $\mathcal{U}$ which acts on the whole symplectic matrix of code blocks. The whole symplectic matrix for $U(A)$ is $U(A) \oplus (U(A))^{-T}$ and is equal to $\mathcal{D}(I_{2k'} + E_{12} \otimes A)$.
    Since upper triangular matrices have the property that $U(A_1) U(A_2) = U(A_1 + A_2)$, we have
    \begin{align}
        \prod_{j=1}^J U(A_j)
        = U\left(\sum_{j=1}^J A_j\right).
    \end{align}
    This implies from Eq.~\eqref{eq:gamma-span} that the multiplication of diagonal gates $\mathcal{D}_{12}(U(A_j))$ realizes $\mathcal{D}_{12}(U(M))$ for arbitrary $M \in M_{k'}(\mathbb{F}_2)$. From the definition of $\rho_{\max}$ in Eq.~\eqref{eq:def_rho_max}, we have $J \leq \rho_{\max}$, and since the subsequent diagonal gates can be combined, i.e., ${\cal D}_1(A^{-1}){\cal D}_1(B)={\cal D}_1(A^{-1}B)$, the total gate count is given by $2\rho_{\max}+1$ from Eq.~\eqref{eq:upper-triangular-decom}.

    In a similar way, a left triangular matrix $L(M)$ for arbitrary matrix $M \in M_{k'}(\mathbb{F}_2)$ can be constructed with the decomposition 
    \begin{align}
        \begin{pmatrix}
            I & \\ A & I
        \end{pmatrix}
        = \begin{pmatrix}
            I & \\ & A
        \end{pmatrix} \begin{pmatrix}
            I & \\ I & I
        \end{pmatrix} \begin{pmatrix}
            I & \\ & A^{-1}
        \end{pmatrix}.
    \end{align}
    The gate count is also the same in this case.
    
    Now, the only thing left to prove is that $U(M)$ and $L(M')$ generate $\GL_{2k'}(\mathbb{F}_2)$, i.e.,
    \begin{align}
        \langle U(M), L(M'); M, M' \in M_{k'}(\mathbb{F}_2) \rangle = \GL_{2k'}(\mathbb{F}_2).
    \end{align}
    In fact, for $k'\geq 3$, we have the following: $\forall G\in \GL_{2k'}(\F_2)$,
    \begin{equation}
        \begin{split}
        G=L(M_1)U(M_2)L(M_3)U(M_4)L(M_5)U(M_6),
        \end{split} \label{eq:six-block-decomp}
    \end{equation}
    which is proved in Ref.~\cite{Urschel2023}. For completeness, we briefly describe this decomposition in the following for reference. Write
    \begin{equation}
        G=\begin{pmatrix}
            A & B \\ C & D
        \end{pmatrix},
    \end{equation}
    and choose $E\in M_{k'}(\F_2)$ such that $B'=AE + B$ is invertible. Set $D'=CE+D$ and 
    \begin{equation}
        K=B'(D'B'^{-1}A - C). \label{eq:def_K}
    \end{equation}
    The matrix $K$ is invertible since 
    \begin{equation}
        \begin{pmatrix}
            B' & A \\ D' & C
        \end{pmatrix} = 
        G \begin{pmatrix}
            E & 1 \\ 1 & 0
        \end{pmatrix}
    \end{equation}
    is invertible, and $D'B'^{-1}A-C$ is the Schur complement of the top-left block. (Note that since the field has characteristic two, plus and minus coincide.)
    Now, choose $X, Y\in \GL_{k'}(\F_2)$ with $X^{-1}Y^{-1}XY=K$. Such $X$ and $Y$ always exist and can algorithmically be constructed for $k'\geq 3$~\cite{Thompson1961}. Then, one may take
    \begin{align}
    M_1&=D'B'^{-1}
      +B'^{-1}X^{-1}Y^{-1}(I_{k'}-X)-B'^{-1}X^{-1}, \\
    M_2&=XB',\\
    M_3&=B'^{-1}X^{-1}(Y-I_{k'}), \\
    M_4&=Y^{-1}(I_{k'}-X)B',\label{eq:six-block-coefficients}\\
    M_5&=B'^{-1}(A-Y), \\
    M_6&=E.
\end{align}
Equation~\eqref{eq:six-block-decomp} can be checked as follows. First, since $I_{k'} + M_2 M_3 = Y$, we have
\begin{equation}
    L(M_1) U(M_2) L(M_3) = \begin{pmatrix}
        Y & M_2 \\
        M_1 Y + M_3 & I_{k'} + M_1 M_2
    \end{pmatrix}.
\end{equation}
Next, since straightforward calculations give $Y M_4 + M_2= B'$ and $(M_1Y + M_3)M_4 + (I + M_1 M_2)=M_1(YM_4 + M_2)+M_3 M_4 + I = D'$, we have 
\begin{equation}
    L(M_1) U(M_2) L(M_3) U(M_4) = \begin{pmatrix}
        Y  & B' \\ M_1 Y + M_3 & D'
    \end{pmatrix}.
\end{equation}
Third, since $Y+B'M_5=A$ and $M_1Y+M_3 =D'B'^{-1}Y - B'^{-1}K$, we have 
\begin{equation}
\begin{split}
    &L(M_1) U(M_2) L(M_3) U(M_4) L(M_5) \\
    &= \begin{pmatrix}
        A  & B' \\ D'B'^{-1}A - B'^{-1}K & D'
    \end{pmatrix},
    \end{split}
\end{equation}
Finally, noticing  $D'B'^{-1}A - B'^{-1}K =C$ from the definition of $K$ in Eq.~\eqref{eq:def_K} and the definitions of $B'$ and $D'$, we have
\begin{equation}
    L(M_1) U(M_2) L(M_3) U(M_4) L(M_5) U(M_6) = \begin{pmatrix}
        A & B \\ C & D
    \end{pmatrix},
\end{equation}
which proves Eq.~\eqref{eq:six-block-decomp}. The gate count also follows from this decomposition.
\end{proof}

\begin{rem}
    The realizability of all the diagonal gates with $\mathcal C_{12}, \mathcal C_{21}$ and single block-diagonal gates $\mathcal D_1(A), \mathcal D_2(B)$ also holds for $k'=2$~\cite{Urschel2023}, but the proof strategy using the commutator construction in Ref.~\cite{Thompson1961} does not work in this case.
\end{rem}
\begin{rem}
    Instead of using the six-gate decomposition given above, one can use a gate teleportation as in Ref.~\cite{Malcolm2026}. This reduces the overhead factor by $1/6$, at the cost of an ancilla code block for each data block.
\end{rem}

\paragraph{Two-block shear gates} \label{par:two-block-shear}
Next, we show that the upper and lower shear of $\Sp_{4k'}(\F_2)$ can be performed with the diagonal $\mathcal D(A)$ gates and the logical gates induced by physical fold-transversal gates.
Notice that the condition $\operatorname{span}_{\mathbb F_2}\Gamma=M_{k'}(\F_2)$ does not imply
\begin{equation}
 \operatorname{span}_{\mathbb F_2}
 \{h\tau_{k'} h^{T}:h\in\Gamma\}
 =\operatorname{Sym}_{k'}(\F_2).
\end{equation}
This is the essential property for the direct shear construction used in Ref.~\cite{Malcolm2026}, which the \nameofcode codes do not have. Consequently, the single \nameofcode code block may not have all the logical Clifford gate actions.
Instead, our construction of shearing gates relies on the following Prop.~\ref{prop:symmetric_matrix_decomp}, where arbitrary $Q\in\GL_{2k'}(\F_2)$ is prepared via a gate sequence as in Prop.~\ref{prop:diagonal_generation}.
We have the following result.
\begin{prop} \label{prop:symmetric_matrix_decomp}
    For $k\geq 3$, let $\tau \in \operatorname{Sym}_{k}(\F_2)$ be a nonsingular nonalternating matrix. Then, for any $R\in\operatorname{Sym}_{k}(\F_2)$, we can decompose it as
    \begin{equation}
        R=\sum_{s=1}^{r}Q_s\tau Q_s^{T},\qquad
 Q_s\in\GL_{k}(\F_2),  \label{eq:two-congruence-phase}
    \end{equation}
    with $r\leq2$.
\end{prop}
A proof of this statement is lengthy, so we give it in Appendix~\ref{sec:proof_matrix_decomp}. In particular, by applying this proposition to the case $\tau=\tau_{k'}\oplus\tau_{k'}$, we can see that an arbitrary gate of the form $\mathcal U(B)$ or $\mathcal L(C)$ with $B,C\in \operatorname{Sym}_{2k'}(\F_2)$ can be generated.
In fact, the phase-type fold-transversal gate $\operatorname{S}_{\rm fd, 1}\otimes \operatorname{S}_{\rm fd, 2}$ sandwiched by the gate sequences realizing $\mathcal D(Q)$ and $\mathcal D(Q^{-1})$ for $Q \in \GL_{2k'}(\F_2)$ given in Prop.~\ref{prop:diagonal_generation} leads to the logical action in the form
\begin{align}
    \begin{split}
        &\begin{pmatrix}
        Q & \vline & \\
        \hline 
        & \vline & Q^{-T}
    \end{pmatrix} \begin{pmatrix}
        I_{k'} & & \tau_{k'} & \\
        & I_{k'} & & \tau_{k'} \\
        & & I_{k'} & \\
        & & & I_{k'} 
    \end{pmatrix} \begin{pmatrix}
            Q^{-1} & \vline & \\
            \hline
            & \vline & Q^T
        \end{pmatrix} \\
    &= \begin{pmatrix}
        I_{2k'} & \vline & Q (\tau_{k'} \oplus \tau_{k'}) Q^T \\
        \hline
        & \vline& I_{2k'}
    \end{pmatrix}
    = \mathcal{U}(Q\tau Q^T).
    \end{split} 
\end{align}
Multiplying the above upper shear leads to the sum in Eq.~\eqref{eq:two-congruence-phase} and thus generates arbitrary $\mathcal{U}(R)$ from Prop.~\ref{prop:symmetric_matrix_decomp}. 
Lower shears can be implemented by inserting upper shears between Hadamard-type fold-transversal gates. Indeed, 
\begin{align}
\begin{split}
    &\begin{pmatrix}
        & \tau \\ \tau
    \end{pmatrix} \begin{pmatrix}
        I_{2k'} & \tau Q \tau Q ^T \tau  \\
        & I_{2k'}
    \end{pmatrix}
    \begin{pmatrix}
        & \tau \\ \tau
    \end{pmatrix}
    = \begin{pmatrix}
        I_{2k'} & \\ Q\tau Q^T & I_{2k'}
    \end{pmatrix} \\
    &= \mathcal{L}(Q\tau Q^T),
\end{split}
\end{align}
where we used the identity $\tau^2 = (\tau_{k'} \oplus \tau_{k'})^2 = I_{2k'}$.

The ability to generate an arbitrary $\mathcal U(B)$ and $\mathcal L(C)$ in $\Sp_{4k'}(\F_2)$ is in fact sufficient to generate all the logical Clifford gates, which can be shown in the next section.

\subsubsection{Multi-block Clifford decomposition}
\label{sec:symplectic-block-decomposition}

In this section, we mention how to compile a multi-block Clifford gate into our elementary logical gate gadgets. Although the basic strategy is the same as that in Ref.~\cite{Malcolm2026}, our available elementary gate gadget is far more limited compared to the code considered in Ref.~\cite{Malcolm2026}. Thus, the total overhead behaves differently from Ref.~\cite{Malcolm2026}.

For simplicity, we assume that $b$ is even and all-to-all connectivity across or within any code blocks is allowed in the following.
Any $S\in \Sp_{2m}(\F_2)$ can be decomposed as~\cite{Malcolm2026}
\begin{equation}
S=\mathcal U(E)\mathcal L(X) \mathcal U(Y) \mathcal L(Z) \mathcal U(W),
\end{equation}
where $E, X, Y, Z, W\in\operatorname{Sym}_{m}(\F_2)$, and $E$ has column weight at most one. The role of $\mathcal U(E)$ is to make the top-left quadrant of $\mathcal U(E)^{-1}S$ invertible. From the commutativity, we have $\mathcal U(B)=\prod_i \mathcal U(B_i)$ when $B=\sum_{i}B_i$. Exploiting this, we can decompose an element of $\operatorname{Sym}_{m}(\F_2)$ into $b\times b$ collections of $\operatorname{Sym}_{k'}(\F_2)$ blocks and perform gates that correspond to off-diagonal terms of $X, Y, Z, W$ given above with high parallelism, since they can be implemented by upper- or lower-triangular diagonal $\mathcal D(A)$ gates in Eq.~\eqref{eq:upper-triangular-decom} sandwiched by fold-transversal Hadamard gate acting on one of the blocks.
The remaining diagonal blocks are in-block shear gates, so they are taken into account locally. See Appendix~\ref{sec:multi-block_Clifford_decomp} for more details.

If there is no $\mathcal U(E)$ term (i.e., $S$ has an invertible top-left quadrant), the worst-case number of layers of gate gadgets is given by 
\begin{equation}
    4(2b+34)\rho_{\max} + 12b + \Theta(1), 
\end{equation}
which is proved in Appendix~\ref{sec:multi-block_Clifford_decomp}. The overhead of $\mathcal U(E)$ depends highly on the structure of $E$. If it can be decomposed into off-diagonal blocks, they can be performed in at most $2\rho_{\max}+\Theta(1)$ depth. If there are diagonals, it may be more costly; each diagonal gate $\mathcal D(Q)$ in $\mathcal U_i(E)$ as shown in Prop.~\ref{prop:symmetric_matrix_decomp} requires six block-unitriangular factors in the worst case, which may result in $36\rho_{\max}+\Theta(1)$ depth.

In Appendix~\ref{sec:multi-block_Clifford_decomp}, we also explain an LDU-type decomposition, which may improve the depth in the special case.

\subsection{Overhead analysis for full Clifford synthesis} \label{subsec:overhead_analysis}
Until this point, we have not specified the exact scaling of $\rho_{\max}$, which determines the time overhead. If $\rho_{\max}=O(k')$, the required gate depth mentioned in the previous section has nearly optimal scaling~\cite{Duncan2020, Maslov2022, Malcolm2026} for practical use. However, we can see below that for a constant-rate family of \nameofcode codes, one has $\rho_{\max}=k'^{2-o(1)}$ in general.

As shown in Eq.~\eqref{eq:cardinality_Borel}, the cardinality of the Borel subgroup $|B|$ is given by $q^{\ell(w_o)}\Theta(q^r)$.
When one takes $q\simeq r \gg 1$, the code parameter $[n, k, d]$ of the Steinberg code satisfies $[G:B]=n=\Theta(k)=\Theta(q^{\ell(w_0)})$.
Thus, we have
\begin{equation}
    |G| = \Theta(q^{2\ell(w_o)+r}).
\end{equation}
Since $r\leq \ell(w_o)\leq 2r^2$ holds from the classification of Chevalley groups~\cite{Carter1972}, the above implies that 
\begin{equation}
    k^{2+o(1)} \leq |G| \leq \Theta(k^3).
\end{equation}
As can be seen from the above, $G$ has almost the smallest possible cardinality for the representation to be absolutely irreducible, since a linear basis of $M_{k}(\F_2)$ has cardinality $k^2$.
In fact, the Steinberg representation is the largest among absolutely irreducible representations for a Chevalley group $G$~\cite{Humphreys1987}.
This accounts for the scaling $\rho_{\max}=k'^{2-o(1)}$. Indeed, a simple counting argument implies that $|M_{k'}(\F_2)|=2^{k'^2}$. To represent an arbitrary element of $M_{k'}(\F_2)$ by a linear combination of $\Gamma$ defined in Eq.~\eqref{eq:def_Gamma}, we have
\begin{equation}
    |\Gamma|^{\rho_{\max}} \gtrsim |M_{k'}(\F_2)|,
\end{equation}
where the left-hand side counts the number of choices of $\rho_{\max}$ elements from $\Gamma$ with replacement, where the duplicated elements account for the length of the summand less than $\rho_{\max}$. Using $|\Gamma|=|G|^2$ and $k'=k^2$, we have
\begin{equation}
    \rho_{\max} \gtrsim \frac{\log|M_{k'}(\F_2)|}{\log|G|^2} = \Theta(k'^2/\log k'),
\end{equation}

In general, finding the optimal linear decomposition with $k'^{2-o(1)}$ summands may not be efficient since it is a satisfiability problem with a large size. An efficient, convenient strategy is to form an operator basis with $\Gamma$ and decompose an element of $M_{k'}(\F_2)$ with this basis, which leads to $\rho_{\max}=k'^2$.

For $b$ blocks of $\llbracket n',k',d\rrbracket$ \nameofcode codes with $b\gg 1$, which have $m = bk'$ logical qubits in total, the gate decomposition mentioned in the previous section leads to the scaling of gate depth, assuming $\rho_{\max}=k'^2$, with
\begin{equation}
    \Theta(bk'^2) = \Theta(m k'),
\end{equation}
which thus have additional $k'$ overhead compared to a conventional compilation~\cite{Duncan2020, Maslov2022, Malcolm2026} or the code families with full logical Clifford gates by fold-transversal gates~\cite{Malcolm2026, koh2026, tansuwannont2026, holmes2026}.
To avoid the large time overhead, we thus need to keep the number $k'$ of logical qubits in a code block small.
Ideally, $k'$ should be polylogarithmically small compared to the total number of logical qubits (i.e., $k'=\operatorname{polylog}(m)$) to keep the time overhead small and the code distance sufficiently large. The situation is similar to the scheme considered in Refs.~\cite{gottesman2014, Tamiya2026} for a constant-space-overhead quantum computation with qLDPC codes. Note again that the \nameofcode code family is not strictly qLDPC.

In terms of the space overhead, our code achieves a constant rate, which is a distinctive feature of the other codes that allow all the logical Clifford gates with transversal and fold-transversal gates~\cite{Malcolm2026, koh2026, tansuwannont2026, holmes2026, albert2026}. Although the \nameofcode code family has sublogarithmically growing stabilizer weights, which means that additional space overhead may be incurred for fault-tolerant state preparation in the asymptotic sense, none of the known codes except those in Ref.~\cite{Malcolm2026} is a qLDPC stabilizer code nor has polylogarithmic stabilizer weight, which may thus require more space overhead for fault-tolerant state preparation. 
The codes developed in Ref.~\cite{Malcolm2026} have only logarithmically many logical qubits.
Thus, the space overhead is a distinct advantage of the \nameofcode code family over the families of codes that have full logical Clifford gates induced by transversal and fold-transversal gates.

\section{Numerical simulation}
\label{sec:numerical-study}
In this section, we perform circuit-level memory simulations for $\llbracket 637, 64, 6 \rrbracket $ and $ \llbracket 2925, 256, 8 \rrbracket$ \nameofcode codes whose seed Steinberg codes are $[21, 8, 6]$ and $ [45, 16, 8]$, respectively. Both Steinberg codes have $q=r=2$, resulting in the stabilizer weight $5$ for these \nameofcode codes. The details of these seed Steinberg codes can be found in Appendix~\ref{appendix:code-examples}. 

The simulations are conducted with the following steps: 1) Initialization of the $\ket{0}^{\otimes n}$ state, 2) 10 rounds of syndrome extraction, and 3) Pauli $Z$ measurements on all the qubits.
The syndrome extraction circuits are constructed so that, for each stabilizer generator, sequential CNOT gates are applied between the data qubits in the support of the generator and a bare ancilla qubit. The depth of the entire syndrome extraction circuit is minimized by partitioning CNOT gates according to the standard edge-coloring problem~\cite{vizing1964estimate}. We have not optimized the CNOT ordering to suppress hook errors.

The standard circuit-level noise is assumed; initialization of physical qubits, physical single-qubit and two-qubit gates, and physical Pauli measurements are all assumed to be noisy and to be subject to the same error rate. Thus, there is only one physical error parameter. 
The syndrome data is generated through implementation of these circuits by the package $\texttt{Stim}$~\cite{Gidney_2021} in Python. 

For decoding, the belief propagation (BP) decoder with the LSD post-processing~\cite{Hillmann_2025} is used. 
We adopt the overlapping window version of the BP-LSD decoder that is implemented in the package \texttt{LDPC}~\cite{Roffe_LDPC_Python_tools_2022, roffe_decoding_2020}. The window size is set equal to the code distance, and the size of the commit region is set to 1. For the parameters of the BP-LSD decoder, the max BP iteration number is set to 30, and the order of LSD is set to 0. 

The memory experiments are executed for a given number of shots, and a shot is counted as failed if at least one of the logical operators is flipped. The number of shots is set to $2 \times 10^{5}$ for the $\llbracket 637, 64, 6 \rrbracket$ \nameofcode code and to $3 \times 10^{5}$ for the $\llbracket 2925, 256, 8 \rrbracket$ \nameofcode code. To decide the flip of logical operators, the logical Pauli basis given in Eq.~\eqref{eq:lshgp-logical-basis-Z} is used.
The logical error rate $p_L$ is calculated as the ratio of failed shots to total shots. Logical error rate per cycle $p_{L, 1}$ is then evaluated through 
\begin{align}
    p_L = 1 - (1 - p_{L,1})^l,
    \label{eq:logical-error-rate-per-cycle}
\end{align}
or equivalently
\begin{align}
    p_{L,1} = 1 - (1 - p_L)^{1/l},
\end{align}
where $l$ denotes the number of syndrome extraction cycles executed. This equation is derived under the assumption that logical errors occur in each cycle independently and that multiple logical flips of any logical Pauli basis are negligibly small. Note that the same scaling is adopted in~\cite{Bravyi_2024}.

For comparison, we have added the logical error rates of multiple blocks of rotated surface codes, each of which has code distance 5 and 7. The numbers of code blocks are chosen to be 64 and 256, respectively, to be consistent with the \nameofcode codes with similar code distances. For the decoding of the rotated surface code, the minimum-weight matching decoder implemented in $\texttt{pymatching}$~\cite{Higgott_2025} is used. For both rotated surface codes with distance 5 and 7, the number of shots is set to $10^9$. 
To scale the value of the logical error rate per round $p_{L,1}$ of a single-block surface code to that of $m$ blocks of surface code, we follow the same procedure in Eq.~\eqref{eq:logical-error-rate-per-cycle}: we use
\begin{align}
    p^{(m)}_{L,1} = 1 - (1 - p_{L,1})^m,
\end{align}
which holds under the assumption that logical errors occur in each block independently.

We use the 95\% Clopper-Pearson interval~\cite{CLOPPER_1934} as the error bar in our plot. Note that the effect of the scaling of logical error rates $p_L$ into $f(p_L)$ by a monotone function $f$ is reflected through the transformation of the original interval $[L, U]$ into $[f(L), f(U)]$.

Fig.~\ref{fig:memory-experiment} shows the result of our numerical simulations.
It demonstrates that the logical error rate of $\llbracket 637, 64, 6 \rrbracket$ \nameofcode codes steadily decreases as the physical error rate decreases, and its performance is better than the 64-block surface code with distance 5; the logical error rate stays lower and decays faster. 
To encode 64 logical qubits, the distance-5 rotated surface code needs 1600 physical qubits while the \nameofcode code needs only 637 physical qubits, an approximately 2.5-fold reduction. 
A similar behavior can be observed for the $\llbracket 2925, 256, 8 \rrbracket$ \nameofcode code compared with the 256-block rotated surface code with distance 7, which is an approximately 4.3-fold reduction in the physical-qubit overhead in this case. Notably, we observe an order-of-magnitude improvement in the logical error rate at the physical error rate $10^{-3}$.
These two plots for the \nameofcode codes also show that there is no obvious hook-error floor~\cite{Dennis_2002}, even though the CNOT ordering was not optimized to avoid hook errors.
We finally note that these results of the \nameofcode codes are obtained with a modest number of maximum BP iterations and the lowest LSD order, highlighting the error-correcting potential of the \nameofcode codes.
This may be owed to the low-weight stabilizers of these codes (in this case 5).

\begin{figure*}[ht]
    \centering
    \includegraphics[width=\linewidth]{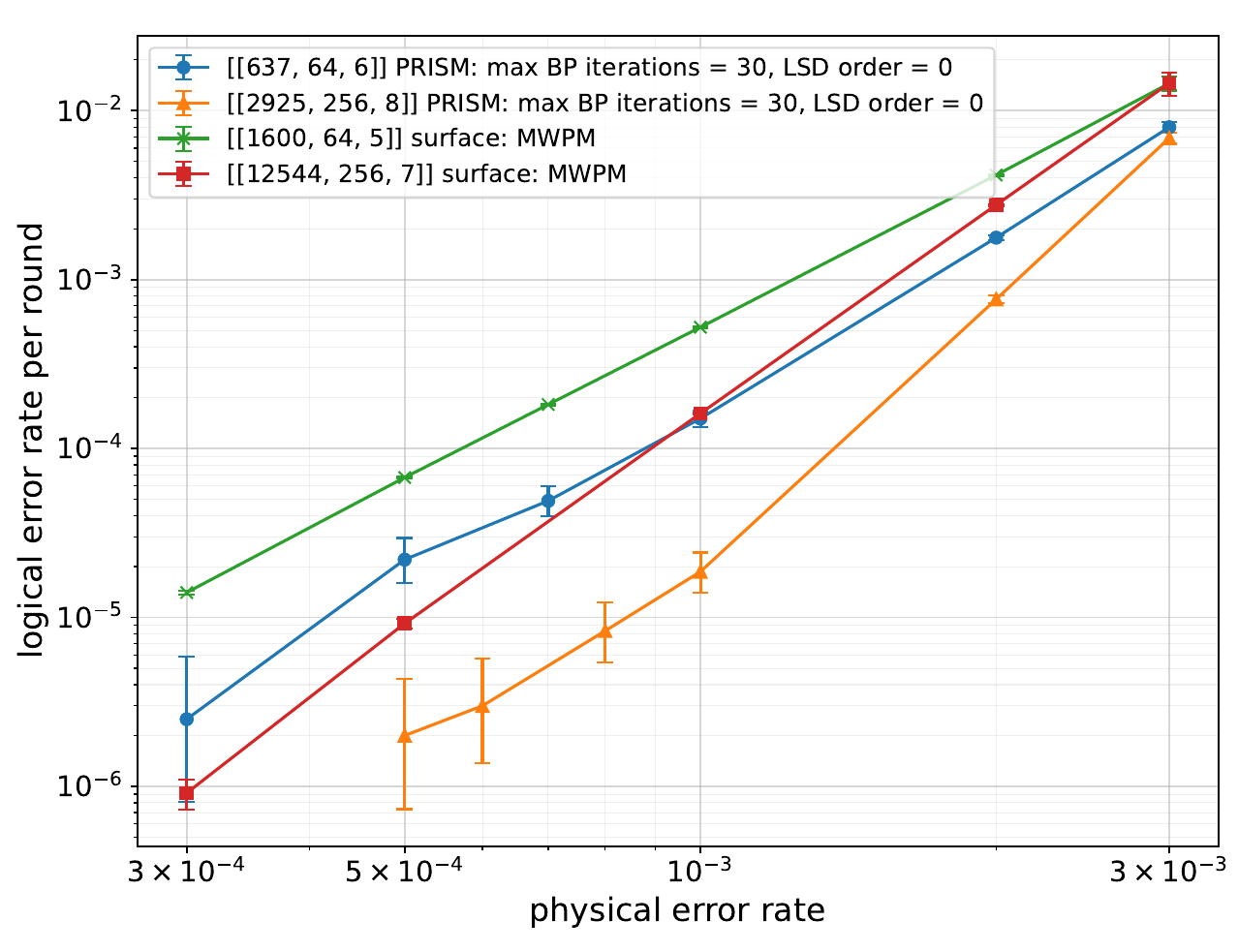}
    \caption{Simulation result of the block logical error rate per syndrome extraction cycle for the memory experiment of the $\llbracket 637, 64, 6 \rrbracket$ and  $\llbracket 2925, 256, 8 \rrbracket$ \nameofcode codes constructed from the seed $[21, 8, 6]$ and $[45, 16, 8]$ Steinberg codes.}
    \label{fig:memory-experiment}
\end{figure*}

\section{Discussion} \label{sec:discussion}

In summary, we have constructed a family of codes named \nameofcode codes with full logical Clifford actions only by physical transversal and fold-transversal gates that simultaneously achieve an asymptotic constant rate, growing distance, and sublogarithmic growth of the stabilizer weight. Our result employs the $ZX$-duality of hypergraph product codes and code automorphisms that are lifted from seed Steinberg codes. The crucial property of a Steinberg code constructed from a Chevalley group over $\F_{2^f}$ is that it has an absolutely irreducible binary representation of the Tanner-graph automorphism group called the Steinberg representation. Thus, qubit permutations of a \nameofcode code block induce logical CNOT-type gates corresponding to the (tensor product of) Steinberg representation. Since an absolutely irreducible representation of a group can generate an arbitrary matrix with a linear combination of its elements, these code automorphisms combined with other transversal and fold-transversal gates can give a generating set of symplectic matrices, which thus implies that full logical Clifford gates can be generated by these gate primitives. 
Our analysis shows that this achievement comes at the cost of longer gate depth that scales as $\Theta(bk'^2)$ for synthesizing an arbitrary logical Clifford gate on $b$ \nameofcode code blocks, each of which encodes $k'$ logical qubits.
Furthermore, we have numerically demonstrated that \nameofcode codes show a steady decrease in logical error rates under a circuit-level noise simulation even with a modest number of BP iterations for the BPLSD decoder. This may be owed to a low stabilizer weight of \nameofcode codes in small instances.

There are several open questions. 
One immediate question is how to inject the magic-state injection into a \nameofcode code. Since the stabilizer weight of the \nameofcode code family grows slowly, one may use the adapter~\cite{cross2025, Swaroop2026}, which is developed for performing a logical multi-qubit Pauli measurement between different qLDPC codes. By connecting a \nameofcode code block to a magic-state factory using the adapter, one can perform a non-Clifford gate, leading to a universal fault-tolerant quantum computation. Whether magic states can more directly be injected into a \nameofcode code block and be distilled by its Clifford gates is open.

Another question is how to effectively exploit the meta-check structure of the \nameofcode codes in decoders. The \nameofcode codes naturally have meta-check structures as shown in Rem.~\ref{rem:meta-check}, and this provides additional constraints that could be incorporated explicitly into a decoder, whereas our present simulations apply BP-LSD directly and therefore do not exploit it. One natural approach would be a hierarchical decoder: the meta-check syndrome could first be used to identify and repair inconsistent syndrome bits, after which BP-LSD would decode the resulting syndrome. Alternatively, a unified decoder could simultaneously infer data and syndrome errors while enforcing both the stabilizer and meta-check constraints. Since the distance of the meta-check code is substantially smaller than that of the code itself, single-shot decoding may not be possible for the \nameofcode code as the code size increases. Nevertheless, exploiting the meta-check constraints may still yield appreciable finite-size improvements under the circuit-level noise, and comparing such structure-aware decoders with the direct BP-LSD baseline is an important direction for improving the decoder.

Going beyond our code construction, an open question is whether an asymptotic constant rate is compatible with a linear gate depth in the number of logical qubits to implement an arbitrary logical Clifford gate. As mentioned in Sec.~\ref{subsec:overhead_analysis}, the \nameofcode codes do not achieve a linear gate depth, unlike Ref.~\cite{Malcolm2026}, which uses a similar compiling strategy. The reason in this case is clear; the cardinality of the automorphism group is too small. To achieve a linear gate depth, the cardinality of the automorphism group needs to be at least exponentially large in the number of logical qubits. Thus, simultaneously achieving a constant rate and a linear-gate-depth logical Clifford action requires an exponentially large code automorphism group in the number of physical qubits. We are not currently aware of such constructions.

Closely related to the above question, our final question is a general tradeoff relation between the applicable logical Clifford gates, the encoding rate, and the stabilizer weight. This has been studied for the phantom-code family~\cite{morris2026}, but the constraint may be extended to more general tradeoffs of code parameters for code families that afford full logical Clifford action via transversal and fold-transversal gates. 
It is expected that parity-check matrices of an LDPC code family do not have rapidly growing non-abelian Tanner-graph automorphism groups that act nontrivially on the code space. This implies that the construction based on the hypergraph product of highly symmetric classical codes, as is done in this work, may have an inherent limitation and may not be able to yield a constant-rate qLDPC family that has full logical Clifford action via transversal and fold-transversal gates. Going beyond a current tradeoff between the encoding rate and the stabilizer weight while keeping full-Clifford applicability may thus require a new technique. Answering this question may be a step towards a deeper understanding of the difference between memory and computation in quantum mechanics.

\section{acknowledgment}
This work is supported by JST, Moonshot R\&D, Grant Number JPMJMS256I.
Y.~C.~is supported by JST, Moonshot
R\&D, Grant Number JPMJMS256E.

\bibliography{ref}% Produces the bibliography via BibTeX.

\appendix

\section{Proof of Proposition~\ref{prop:symmetric_matrix_decomp}} \label{sec:proof_matrix_decomp}
In this section, we prove Prop.~\ref{prop:symmetric_matrix_decomp}.
A symmetric matrix $A$ is called \emph{alternating} when $xAx^T=0$ for every $x$.  Over $\F_2$ this is equivalent to $\diag(A)=0$.  Thus, $A$ is \emph{nonalternating} precisely when at least one of its diagonal entries is nonzero. We first restate Prop.~\ref{prop:symmetric_matrix_decomp}.

\begin{prop}
\label{prop:two-congruence}
Let $k\geq3$. Then, for every $R\in\Sym_k(\F_2)$, there exist
$Q_1,Q_2\in\GL_k(2)$ such that
\begin{equation}
 R=Q_1\tau Q_1^{T}+Q_2\tau Q_2^{T}.
 \label{eq:target-decomposition}
\end{equation}
\end{prop}

We first state the two ingredients used in the proof of the above proposition. 

\begin{lem}[Binary symmetric normal form~\cite{Buckhiester1973}]
\label{lem:normal-form}
Every symmetric $R\in M_k(\F_2)$ is congruent to one of the forms
\begin{equation}
 I_r\oplus0_z,
 \qquad r+z=k,
 \label{eq:nonalternating-normal-form}
\end{equation}
or
\begin{equation}
 X^{\oplus h}\oplus0_z,
 \qquad
 X=\begin{pmatrix}0&1\\1&0\end{pmatrix},
 \qquad 2h+z=k.
 \label{eq:alternating-normal-form}
\end{equation}
The second case is the alternating case.  In particular, every nonsingular nonalternating symmetric $k\times k$ matrix is congruent to $I_k$.
\end{lem}
We omit the proof of this and refer the reader to Ref.~\cite{Buckhiester1973}.

\begin{lem}[Splitting the normal forms]
\label{lem:normal-form-splitting}
Let $C$ be either normal form in
\eqref{eq:nonalternating-normal-form}--\eqref{eq:alternating-normal-form},
of total dimension $k\geq3$.  Then
\begin{equation}
 C=X+Y
\end{equation}
for symmetric, nonsingular, nonalternating matrices $X,Y\in M_k(\F_2)$.
\end{lem}

\begin{proof}
We treat the nonalternating and alternating normal forms separately.

\smallskip
\noindent\emph{Case 1: $C=I_r\oplus0_z$.}
For dimensions two and three, set
\begin{align}
 A_2&=\begin{pmatrix}1&1\\1&0\end{pmatrix},
 &B_2=A_2+I_2
 &=\begin{pmatrix}0&1\\1&1\end{pmatrix},
 \label{eq:I2-splitting}\\[1ex]
 A_3&=\begin{pmatrix}
 0&1&1\\
 1&1&0\\
 1&0&0
 \end{pmatrix},
 &B_3=A_3+I_3
 &=\begin{pmatrix}
 1&1&1\\
 1&0&0\\
 1&0&1
 \end{pmatrix}.
 \label{eq:I3-splitting}
\end{align}
All four matrices are symmetric, nonsingular, and nonalternating.  Every integer $r\geq2$ has the form $r=2a+3b$ with $a,b\geq0$.  Direct sums of
\eqref{eq:I2-splitting} and \eqref{eq:I3-splitting} therefore give
\begin{equation}
 I_r=A_r+B_r,
\end{equation}
where $A_r$ and $B_r$ have all three required properties.  The radical is then added by
\begin{equation}
 I_r\oplus0_z
   =(A_r\oplus I_z)+(B_r\oplus I_z).
 \label{eq:append-zero-block}
\end{equation}

It remains to treat $r=0$ and $r=1$.  If $r=0$, then simply
$0_k=I_k+I_k$.  If $r=1$, the assumption $k\geq3$ implies $z\geq2$, and
\begin{equation}
 \begin{pmatrix}
 1&0&0\\
 0&0&0\\
 0&0&0
 \end{pmatrix}
 =
 \begin{pmatrix}
 0&1&0\\
 1&0&0\\
 0&0&1
 \end{pmatrix}
 +
 \begin{pmatrix}
 1&1&0\\
 1&0&0\\
 0&0&1
 \end{pmatrix}.
 \label{eq:rank-one-splitting}
\end{equation}
The two matrices on the right have determinant one and nonzero diagonal. Any remaining zero coordinates are appended as in
\eqref{eq:append-zero-block}.

\smallskip
\noindent\emph{Case 2: $C=H^{\oplus h}\oplus0_z$.}
The case $h=0$ was already handled by $0_k=I_k+I_k$.  If $h=1$, then $k\geq3$ forces $z\geq1$, and
\begin{equation}
 \begin{pmatrix}
 0&1&0\\
 1&0&0\\
 0&0&0
 \end{pmatrix}
 =
 \begin{pmatrix}
 0&0&1\\
 0&1&0\\
 1&0&0
 \end{pmatrix}
 +
 \begin{pmatrix}
 0&1&1\\
 1&1&0\\
 1&0&0
 \end{pmatrix}.
 \label{eq:H-plus-zero-splitting}
\end{equation}
Again, the two summands are symmetric, have determinant one, and have nonzero diagonals.  Extra zero coordinates are appended using identical identity blocks.

Now suppose $h\geq2$.  After a simultaneous permutation of rows and columns,
\begin{equation}
 H^{\oplus h}\sim
\Omega_h:=\begin{pmatrix}0&I_h\\I_h&0\end{pmatrix},
\end{equation}
where $\sim$ denotes the equivalence up to row and column permutations.
Choose $P_h\in\GL_h(\F_2)$ such that $P_h+I_h$ is also invertible. Such a choice exists for every $h\geq2$: write $h=2a+3b$ and take an appropriate direct sum of
\begin{equation}
 P_2=\begin{pmatrix}0&1\\1&1\end{pmatrix},
 \qquad
 P_3=\begin{pmatrix}
 0&0&1\\
 1&0&1\\
 0&1&0
 \end{pmatrix}.
 \label{eq:P2-P3}
\end{equation}
Both $P_j$ and $P_j+I_j$ are invertible for $j=2,3$.  Define
\begin{align}
 X_h&=\begin{pmatrix}I_h&P_h\\P_h^{T}&0\end{pmatrix},
 &Y_h&=\begin{pmatrix}
 I_h&P_h+I_h\\
 P_h^{T}+I_h&0
 \end{pmatrix}.
 \label{eq:Hh-splitting}
\end{align}
Then $X_h+Y_h=\Omega_h$.  The matrices are symmetric and nonalternating because their upper-left block is $I_h$.  To prove that $X_h$ is invertible, suppose
\begin{equation}
 X_h\binom{x}{y}=0.
\end{equation}
The lower block equation gives $P_h^{T}x=0$, hence $x=0$; the upper block equation then gives $P_h y=0$, hence $y=0$.  The same argument, with $P_h+I_h$ in place of $P_h$, proves that $Y_h$ is invertible.  Appending
$I_z$ to both summands incorporates the radical $0_z$.  Finally, undoing the coordinate permutation transports this splitting from $\Omega_h$ to $H^{\oplus h}$.
\end{proof}

\begin{proof}[Proof of Proposition~\ref{prop:symmetric_matrix_decomp}]
By Lemma~\ref{lem:normal-form}, choose $T\in\GL_k(\F_2)$ such that
\begin{equation}
 TRT^{T}=C
\end{equation}
is one of the two normal forms.  By Lemma~\ref{lem:normal-form-splitting}, write $C=X+Y$, where $X$ and $Y$ are symmetric, nonsingular, and nonalternating.  Hence
\begin{equation}
 R=X'+Y',
 \qquad
 X'=T^{-1}XT^{-T},
 \qquad
 Y'=T^{-1}YT^{-T}.
 \label{eq:transport-back}
\end{equation}
Congruence preserves symmetry, nonsingularity, and the alternating versus nonalternating property, so $X'$ and $Y'$ are again nonsingular nonalternating symmetric matrices.

Lemma~\ref{lem:normal-form} shows that $X'$, $Y'$, and $\tau$ are all congruent to $I_k$.  Therefore $X'$ and $Y'$ are each congruent to $\tau$: there exist $Q_1,Q_2\in\GL_k(\F_2)$ such that
\begin{equation}
 X'=Q_1\tau Q_1^{T},
 \qquad
 Y'=Q_2\tau Q_2^{T}.
\end{equation}
Substitution into \eqref{eq:transport-back} gives
\eqref{eq:target-decomposition}.

The proof is constructive.  For example, if congruence reduction produces
matrices $U_\tau,U_1,U_2$ satisfying
\begin{equation}
 U_\tau\tau U_\tau^{T}=I_k,
 \qquad
 U_1X'U_1^{T}=I_k,
 \qquad
 U_2Y'U_2^{T}=I_k,
\end{equation}
then we have
\begin{equation}
 Q_1=U_1^{-1}U_\tau,
 \qquad
 Q_2=U_2^{-1}U_\tau.
\end{equation}
\end{proof}

\begin{rem}[Relation to Ref.~\cite{Malcolm2026}]
Supplementary Lemma~6 of Ref.~\cite{Malcolm2026} proves that every $M\in M_k(\F_2)$ is a sum of at most two elements of $\GL_k(\F_2)$.  Its $I_2$ construction agrees with \eqref{eq:I2-splitting}.  That lemma does not, however, require its two summands to be symmetric or nonalternating,
and its general rank reduction uses independent row and column operations $gMg'$ rather than a congruence $hRh^{T}$.  In particular, the displayed $I_3$ summands in its Eq.~(55) are not symmetric, and one summand in its rank-one Eq.~(58) is alternating.  Equations \eqref{eq:I3-splitting}, \eqref{eq:rank-one-splitting}, and \eqref{eq:H-plus-zero-splitting} provide the additional symmetry/nonalternation refinement required for Prop~\ref{prop:symmetric_matrix_decomp}.
\end{rem}

\section{Constructive multi-block Clifford decomposition} \label{sec:multi-block_Clifford_decomp}

Let the number $b$ of code blocks be even for simplicity, the number $k'$ of logical qubits in each code block be $k'\geq 3$, and put $m=bk'$.  
As explained in Sec.~\ref{subsec:Clifford-synthesis} in the main text, the local diagonal gate $\mathcal D_i(G)$ is defined as
\begin{equation}
    \mathcal D_i(G) = \begin{pmatrix}
        \begin{matrix}
            I_{(i-1)k'} &  & O  \\
             & G &  \\
            O & & I_{(b-i)k'}
        \end{matrix} & \vline &  \hspace{0.5em} O \hspace{0.5em} \\ \hline
        \hspace{0.5em} O \hspace{0.5em} & \vline & \begin{matrix}
            I_{(i-1)k'} &  & O  \\
             & G^{-T} &  \\
            O & & I_{(b-i)k'}
        \end{matrix}
    \end{pmatrix}.
\end{equation}
Similarly, the $i$-th-block local upper shear $\mathcal U_i(R)$ with a symmetric $R\in M_{k'}(\F_2)$ is defined as
\begin{equation}
\mathcal U_i(R) = 
\begin{array}{c @{\hspace{-0.7em}} l}
   \begin{matrix} 
   \begin{matrix}\hspace{1.4cm} i\end{matrix} 
   \\
  \begin{pmatrix}
    \begin{matrix} & & \\ & I_m & \\ & & \end{matrix} 
    & \vline & 
    \begin{matrix} 0 & \cdots & 0 \\ \vdots & R & \vdots \\ 0 & \cdots & 0 \end{matrix} \\
    \hline
    \begin{matrix} & & \\ & O & \\ & & \end{matrix} 
    & \vline & 
    \begin{matrix} & & \\ & I_m & \\ & & \end{matrix}
  \end{pmatrix}.
  \end{matrix}
  & 
  \begin{matrix} 
  \\
    \\[-1em]
    i \\
    \\[2em]
  \end{matrix}
\end{array}
\end{equation}
The matrix $\operatorname{S}_{\tau_{k'}, i} $, which can be implemented by a phase-type fold-transversal gate on the $i$-th block, can then be written as $\operatorname{S}_{\tau_{k'}, i}=\mathcal U_i(\tau_{k'})$.

As mentioned in Sec.~\ref{subsec:Clifford-synthesis}, by combining elementary gate gadgets implementing $\mathcal D_i(A)$ and $\mathcal C_{ij}$, one can implement the following gate for an arbitrary $M\in M_{k'}(\F_2)$:
\begin{equation}
    \mathcal C_{ij}(M) \coloneqq \mathcal D(I_m+E_{ij}\otimes M). \label{eq:arbitrary_mat_decomp}
\end{equation}
Note that $\mathcal C_{ij}(I_{k'})=\mathcal C_{ij}$.
Since $E_{ij}E_{ij}=0$ for $i\neq j$, this gate is self-inverse.  
Furthermore, since upper-triangular matrices with the same ordered support add their coefficients, we have
\begin{align}
 \mathcal C_{ij}(M)
 &=\prod_{s=1}^{\rho_\Gamma(M)}\mathcal C_{ij}(h_s),
 \label{eq:C-addition}\\
 \mathcal C_{ij}(h_s)
 &=\mathcal D_i(h_s)\,\mathcal C_{ij}\,
   \mathcal D_i(h_s)^{-1}.
 \label{eq:C-conjugation}
\end{align}
Expanding $\mathcal{C}_{ij}(M)$ with Eq.~\eqref{eq:C-conjugation}, we can merge consecutive diagonal gates as follows.
\begin{align}
    \mathcal{C}_{ij}(M) = \mathcal{D}_i(h_1) \mathcal{C}_{ij} \mathcal{D}_i(h_1^{-1} h_2) 
    \cdots
    \mathcal{C}_{ij} \mathcal{D}_i(h_{\rho_\Gamma(M)}^{-1}).
    \label{eq:merge-adjacent-diagonal-gates}
\end{align}
Thus, for an arbitrary $M \in M_{k'}(\F_2)$, $\mathcal{C}_{ij}(M)$ requires at most
$\rho_{\max}\coloneqq \max_M\rho_\Gamma(M)\leq k'^2$ uses of the two-block primitive gates $\mathcal C_{ij}$ and the single-block primitive code-automorphism gadgets, which sums up to $2\rho_{\max}+1$ gate gadgets, as described in the proof of Prop.~\ref{prop:diagonal_generation}.

\subsection{Global LULU/LDU decomposition}
\label{subsec:global-gate-decom}

Now, we aim to decompose a general element
\begin{equation}
 S=\begin{pmatrix}A&B\\C&D\end{pmatrix}\in\Sp_{2m}(\F_2).
\end{equation}
As has been shown in Ref.~\cite{Malcolm2026}, when the top-left quadrant $A$ is invertible, we can decompose it as~\cite[Lemma~32]{Malcolm2026}
\begin{equation}
    S=\mathcal L(X) \mathcal U(Y) \mathcal L(Z) \mathcal U(W).\label{eq:alternative_LULU}
\end{equation}
We call it the LULU decomposition.
Another possibly useful decomposition is~\cite[Lemma~33]{Malcolm2026}
\begin{equation}
 S=\mathcal L(Y)\,\mathcal D(A)\,\mathcal U(Z).
 \label{eq:invertible-A-LDU}
\end{equation}
The former decomposition in general leads to less overhead, as can be seen in the following, but the overhead depends on the balance between the required constant in the decomposition, the number $b$ of code blocks, and the number $k'$ of logical qubits in a code block.

For singular $A$, there inevitably appears another $\mathcal U(E)$ with $E\in\Sym_{m}(\F_2)$ and the row weight of $E$ at most $1$, such that~\cite{Malcolm2026}
\begin{equation}
    S=\mathcal U(E)\,S',
 \label{eq:general-A-LDU}
\end{equation}
where $S'$ has an invertible top-left quadrant.
Thus, an additional overhead to implement $\mathcal U(E)$ needs to be taken into account.
It thus remains to synthesize global symmetric shears $\mathcal U$ and $\mathcal L$, and possibly one $\mathcal D$ factor as in Eq.~\eqref{eq:invertible-A-LDU} or \eqref{eq:alternative_LULU}.

\subsection{Multi-block shears}
\label{subsec:shear-synthesis}

Partition a symmetric matrix $B\in \Sym_m(\F_2)$ into $b^2$ number of $k'\times k'$ blocks $B_{ij}$.
We separate the discussion of off-diagonal terms $B_{ij}$ for $i \neq j$ and diagonal terms $B_{ii}$.

\subsubsection{Inter-block shears} \label{subsubsec:inter-block-shear}
We process inter-block shears, i.e., off-diagonal terms $B_{ij}$, into layers of two-block shears. For $i<j$, define
\begin{equation}
 \mathcal U_{ij}(B)
 :=\mathcal U(E_{ij}\otimes B_{ij}+E_{ji}\otimes B_{ji}),
\end{equation}
with $B_{ji}=B_{ij}^T$.
This is a CZ-type circuit and can thus be synthesized as
\begin{equation}
 \mathcal U_{ij}(B)=\operatorname{H}_{\tau_{k'},j} \mathcal C_{ij}(B_{ij}\tau_{k'})\operatorname{H}_{\tau_{k'},j}.
 \label{eq:offdiagonal-shear-gadget}
\end{equation}
Indeed, since the inverse of upper right triangular matrix $U(A)$ is $U(A)$ itself and $\tau = \tau^{-1}$ holds for an involutive permutation $\tau = \tau_{k'}$, the direct calculation shows that
\begin{align}
\begin{split}
    &\begin{pmatrix}
        \begin{matrix}
            I & \\ & 0
        \end{matrix} &\vline&  
        \begin{matrix}
           0 & \\ & \tau
        \end{matrix} \\
        \hline
    \begin{matrix}
       0 & \\  & \tau
    \end{matrix} & \vline &
    \begin{matrix}
        I & \\ & 0
    \end{matrix}
    \end{pmatrix} 
    \begin{pmatrix}
        \begin{matrix}
            I & B_{ij} \tau \\
            & I
        \end{matrix} & \vline & 
        \begin{matrix}
            &  \\ &
        \end{matrix} \\
        \hline
        \begin{matrix}
            & \\ &
        \end{matrix} & \vline & 
        \begin{matrix}
            I & \\ \tau B_{ij}^T & I
        \end{matrix}
    \end{pmatrix}
    \begin{pmatrix}
        \begin{matrix}
            I & \\ & 0
        \end{matrix} &\vline&  
        \begin{matrix}
           0 & \\ & \tau
        \end{matrix} \\
        \hline
    \begin{matrix}
        0 & \\ & \tau
    \end{matrix} & \vline &
    \begin{matrix}
        I & \\ & 0
    \end{matrix}
    \end{pmatrix} \\
    &= \begin{pmatrix}
        \begin{matrix}
            I & B_{ij} \tau \\
            & 
        \end{matrix} & \vline & 
        \begin{matrix}
            & \\  B_{ij}^T & \tau
        \end{matrix} \\
        \hline 
        \begin{matrix}
            0 & \\ & \tau
        \end{matrix} & \vline & 
        \begin{matrix}
            I & \\ & 0
        \end{matrix}
    \end{pmatrix}
    \begin{pmatrix}
        \begin{matrix}
            I & \\ & 0
        \end{matrix} &\vline&  
        \begin{matrix}
           0 & \\ & \tau
        \end{matrix} \\
        \hline
    \begin{matrix}
        0 & \\ & \tau
    \end{matrix} & \vline &
    \begin{matrix}
        I & \\ & 0
    \end{matrix}
    \end{pmatrix} \\
    &= \begin{pmatrix}
        I & &\vline&  & B_{ij} \\
        & I  &\vline&  B_{ij}^T & \\ \hline
        0 &  &\vline& I & \\
        & 0 & \vline& & I
    \end{pmatrix}. 
\end{split}
\end{align}
From~\eqref{eq:offdiagonal-shear-gadget}, the gate depth for two-block off-diagonal terms is thus given by $2\rho_{\max} + 3$.

Because $b$ is even, the edges of the complete graph $K_b$ split into
$b-1$ perfect matchings.  Hence, all the off-diagonal gadgets $\prod_{i<j}\mathcal U_{ij}(B)$ can be organized in $b-1$ rounds of disjoint block pairs.  This is the same edge-coloring argument used for multi-block diagonal operations in Supplementary Lemma~26 of Ref.~\cite{Malcolm2026}.

Similarly, the off-diagonal blocks of the lower shear $\mathcal L(C)$ can be obtained by sandwiching $\mathcal C_{ij}(\tau_{k'}C)$ with $\operatorname{H}_{\tau_{k'},i} $ as 
\begin{equation}
 \mathcal L(C)
 =\operatorname{H}_{\tau_{k'},i}\,\mathcal C_{ij}(\tau_{k'}C)\,\operatorname{H}_{\tau_{k'},i},
 \label{eq:lower-from-upper}
\end{equation}
for each $i<j$.
The gate depth of it is thus given by $2\rho_\text{max} + 3$.

Followed by the diagonal term construction in the next section, this produces $\mathcal U(B)$ exactly.

\subsubsection{In-block shears} \label{subsubsec:in-block-shear}
\label{subsec:local-phase-full-gl}

Using Prop.~\ref{prop:symmetric_matrix_decomp}, we can decompose diagonal shear gadgets $\mathcal{U}_i(R)$ as
\begin{equation}
 \mathcal{U}_i(B_{ii})=
 \prod_{s=1}^{r}
 \mathcal D_i(Q_s)\operatorname{S}_{\tau_{k'},i}\mathcal D_i(Q_s)^{-1},
 \label{eq:local-phase-conjugation}
\end{equation}
where $B_{ii}=\sum_{s=1}^r Q_s \tau_{k'} Q_s^{T}$ with $r \leq 2$.
The arbitrary $\mathcal D_i(Q_s)$ in \eqref{eq:local-phase-conjugation} is synthesized with an auxiliary block, not as a word in $\Gamma$, as explained in Prop.~\ref{prop:diagonal_generation}.  
Thus, it is convenient that the two members of a fixed pair $(i,j)$ are processed together through
\begin{equation}
 \mathcal{D}_{ij}(Q_{s_i} \oplus Q_{s_j})
 \in\GL_{2k'}(\F_2).
 \label{eq:paired-local-GL}
\end{equation}
More specifically, two-block paired diagonal terms are produced by
\begin{align}
\begin{split}
    &\mathcal{U}_{i}(B_{ii}) \mathcal{U}_{j}(B_{jj}) \\
    &= \prod_{s=1}^r \mathcal{D}_{ij}(Q_{s_i} \oplus Q_{s_j}) (\operatorname{S}_{\tau_{k'},i} \cdot \operatorname{S}_{\tau_{k'},j}) 
    \mathcal{D}_{ij}(Q_{s_i} \oplus Q_{s_j})^{-1} .
\end{split}
\end{align}
As in Eq.\eqref{eq:merge-adjacent-diagonal-gates}, adjacent diagonal gates can be merged so that at most three, rather than four, diagonal gates and two parallel fold-transversal S gates are applied.

Since each $\mathcal{D}_{ij}(Q_{s_i} \oplus Q_{s_j})$ term is implemented by a six-gate decomposition as in Eq.~\eqref{eq:six-block-decomp} and thus requires $6(2\rho_{\max}+1)$ gate depth as shown in Prop.~\ref{prop:diagonal_generation}, the total gate depth to implement $\mathcal U(B)$ is given by
\begin{equation}
\begin{split}
    &(b-1)(2\rho_{\max}+3) + 3 \cdot 6(2\rho_{\max} + 1) + 2 \\
    &= (2b+34)\rho_{\max} + 3b + \Theta(1).
    \end{split} \label{eq:one-symmetric-shear-upper}
\end{equation}
The same depth count holds for $\mathcal L(C)$.
The total gate depth count for the LULU decomposition Eq.~\eqref{eq:alternative_LULU} is thus given by
\begin{align}
\begin{split}
    d_\text{LULU} 
    &= 4\bigl((2b + 34)\rho_\text{max} + 3b + \Theta(1)\bigr)  \\
    &= 4(2b + 34)\rho_\text{max} + 12b + \Theta(1).
\end{split}
\label{eq:total-cost-LULU}
\end{align}

\subsection{Multi-block diagonal gates by recursive block-$PLDU$}
\label{subsec:levi-compilation}

We now describe how to synthesize a multi-block diagonal gate
\begin{equation}
 \mathcal D(A)=\begin{pmatrix}A&0\\0&A^{-T}\end{pmatrix},
 \qquad A\in\GL_{bk'}(\F_2),
\end{equation}
with parallel depth linear in the number $b$ of code blocks.  
The map $A\mapsto\mathcal D(A)$ preserves products, so it is enough to factor $A$ itself.

The construction has a simple recursive form.  We divide the code blocks into two parts, use lower and upper block-unitriangular matrices to separate them, and then apply the same procedure to the two diagonal pieces in
parallel. We stop when a diagonal piece acts on at most two blocks. Each resulting two-block piece is synthesized directly by the construction of Prop.~\ref{prop:diagonal_generation}, and any remaining one-block remainders are paired and treated by the same construction.  We first explain this construction assuming that every required pivot is invertible.  We then show that one permutation before the recursion can ensure this condition in general.  Together, these steps give the block-$PLDU$ (permutation--lower--diagonal--upper) decomposition.

\subsubsection{One recursive step}

Suppose first that $b=2q$, and divide the blocks into two halves of $q$.
Write
\begin{equation}
A=\begin{pmatrix}A_{11}&A_{12}\\A_{21}&A_{22}\end{pmatrix},
 \qquad A_{ij}\in M_{qk'}(\F_2),~i,j\in[2],
 \label{eq:block-plu-split}
\end{equation}
where $[n]\coloneqq \{1,\ldots,n\}$.
Assume for the moment that $A_{11}$ is invertible, and define
\begin{align}
 \Phi&=A_{21}A_{11}^{-1},\\
 \Theta&=A_{11}^{-1}A_{12},\\
 \Sigma&=A_{22}+A_{21}A_{11}^{-1}A_{12}.
 \label{eq:block-plu-RTS}
\end{align}
Direct multiplication gives the following LDU decomposition:
\begin{equation}
 A=
 \underbrace{\begin{pmatrix}I&0\\ \Phi&I\end{pmatrix}}_{L_\times(\Phi)}
 \underbrace{\begin{pmatrix}A_{11}&0\\0&\Sigma\end{pmatrix}}_{A_{\rm diag}}
 \underbrace{\begin{pmatrix}I&\Theta\\0&I\end{pmatrix}}_{U_\times(\Theta)}.
 \label{eq:recursive-block-LDU}
\end{equation}
Indeed, the right-hand side is
\begin{equation}
 \begin{pmatrix}
 A_{11}&A_{11}\Theta\\
 \Phi A_{11}&\Phi A_{11}\Theta+\Sigma
 \end{pmatrix}
 = \begin{pmatrix}A_{11}&A_{12}\\A_{21}&A_{22}\end{pmatrix}.
\end{equation}
Since both $A$ and $A_{11}$ are invertible, $\Sigma$ is also invertible.
Thus, after implementing the two cross-half factors $L_\times(\Phi)$ and
$U_\times(\Theta)$, we may factor $A_{11}$ and $\Sigma$ recursively and in parallel.

\subsubsection{Implementing the cross-half factors}
Partition $\Phi$ and $\Theta$ into $q\times q$ arrays of $k'\times k'$ matrices,
denoted $\Phi_{ij}$ and $\Theta_{ij}$, where $i,j\in[q]$. Define $T_{ij}(M)$ as
\begin{equation}
 T_{ij}(M)=I_{bk'}+E_{ij}\otimes M, 
\end{equation}
which then satisfies $\mathcal C_{ij}(M)=\mathcal D(T_{ij}(M))$. Then,
\begin{align}
 L_\times(\Phi)
 &=\prod_{i=1}^{q}\prod_{j=1}^{q} T_{q+i,j}(\Phi_{ij}),
 \label{eq:cross-lower-product}\\
 U_\times(\Theta)
 &=\prod_{i=1}^{q}\prod_{j=1}^{q}
   T_{i,q+j}(\Theta_{ij}).
 \label{eq:cross-upper-product}
\end{align}
All factors in each displayed product commute. For example, every factor in Eq.~\eqref{eq:cross-lower-product} adds a control block in the first half to a target block in the second half. Therefore, multiplying two such factors creates
no additional cross term.  The same argument applies to
Eq.~\eqref{eq:cross-upper-product}, with the two halves interchanged.

Each product contains all $q^2$ pairs formed by one block from each half. They can be scheduled in $q$ rounds.  One explicit schedule is to pair block $a$ in the first half with block $q+1+((a+r-2)\bmod q)$ in round $r=1,\ldots,q$.  Every cross-half pair then appears exactly once, and no block appears twice in the same round.  Hence,
$L_\times(\Phi)$ requires $q$ matching rounds, and $U_\times(\Theta)$ requires another $q$ matching rounds.

By Eqs.~\eqref{eq:arbitrary_mat_decomp}--\eqref{eq:C-conjugation}, each $\mathcal C_{ij}(M)$ has depth at most $2\rho_{\max}+1$.  All block pairs in one round are disjoint, so their decompositions run in parallel.  Thus, one recursive step on $2q=b$ blocks contributes
$b(2\rho_{\max}+1)$ to the depth, in addition to the two recursive children, which are also processed in parallel.

\subsubsection{Depth recurrence and remaining factors}
Let $g(b)$ denote the number of cross-half pair matching rounds before expanding each $\mathcal C_{ij}(M)$ and before synthesizing the remaining single-block transformations.  When $b$ is a power of two, the two children of every recursive node
have the same size and are processed in parallel.  Therefore,
\begin{equation} 
\begin{cases}
 g(1)=g(2)=0,\\
 g(b)=g(b/2)+b & (b\geq4).
 \end{cases}
 \label{eq:block-plu-recurrence}
\end{equation}
It thus follows that
\begin{equation}
 g(b)=b+\frac b2+\frac b4+\cdots + 4=2b-4.
 \label{eq:block-plu-rounds}
\end{equation}
For a general even $b$ that is not a power of two, one may instead split the blocks as evenly as possible. At a split into $\lfloor b/2\rfloor$ and $\lceil b/2\rceil$ blocks, the two cross factors require $2\lceil b/2\rceil$ rounds.  The balanced recursion therefore obeys, for $b\geq3$,
\begin{equation}
 g(b)=
 \max\left\{
 g\left(\left\lfloor\frac b2\right\rfloor\right),
 g\left(\left\lceil\frac b2\right\rceil\right)
 \right\}
 +2\left\lceil\frac b2\right\rceil.
\end{equation}
The recurrence is nondecreasing in $b$, so the larger child determines the maximum.  Iterating along that child gives
\begin{equation}
 g(b)
 =2\sum_{s=1}^{\lceil\log_2 b\rceil-1}
 \left\lceil\frac{b}{2^s}\right\rceil
 =2b+O(\log b).
 \label{eq:block-plu-rounds-general}
\end{equation}
The $O(\log b)$ term is the worst-case accumulation of the rounding at each level; it is not incurred for every non-power-of-two value of $b$.  It can, however, occur within this balanced recursion.  For example, if
$b=2^r+2$ with $r\geq2$, then
\begin{equation}
 g(b)=2b+2r-6.
\end{equation}
The dominant term of the round count remains linear in $b$.  Equation~\eqref{eq:block-plu-rounds} is exact for powers of two, whereas Eq.~\eqref{eq:block-plu-rounds-general}
is a worst-case description of the rounding overhead for the balanced recursion; it is not a lower bound for other possible synthesis methods.

After the recursive procedure given above, there remain matrices of the form
\begin{equation}
 A_{\rm rem}
 =\bigoplus_{\ell=1}^{t}G_\ell
  \oplus\bigoplus_{j=1}^{2u}Q_j,
 \label{eq:block-plu-leaves}
\end{equation}
where $G_\ell\in\GL_{2k'}(\F_2)$, $Q_j\in\GL_{k'}(\F_2)$, and $2t+2u=b$.
Here, each $G_\ell$ is a two-block remainder at which the recursion stops.  The number of one-block remainders is even because $b$ is even.  We pair these leaves and implement
\begin{equation}
 \mathcal D(G_1),\ldots,\mathcal D(G_t),
 \quad
 \mathcal D(Q_1\oplus Q_2),\ldots,
 \mathcal D(Q_{2u-1}\oplus Q_{2u})
\end{equation}
on disjoint block pairs in parallel. 
By Eq.~\eqref{eq:six-block-decomp}, every displayed two-block gate has depth at most $6(2\rho_{\max}+1)$, and hence the entire remainder has this same depth.  

Combining the cross-half rounds with the transformations of remainders gives the pivot-free bound
\begin{equation}
 d_{\rm base}(b)
 \leq \bigl(g(b)+6\bigr)(2\rho_{\max}+1).
 \label{eq:block-plu-base-cost}
\end{equation}
In particular, for $b$ a power of two,
\begin{align}
 d_{\rm base}(b)
 &\leq (2b+2)(2\rho_{\max}+1).
 \label{eq:block-plu-no-pivot-cost}
\end{align}

\subsubsection{Ensuring invertible pivots}

The preceding recursion assumes that the top-left part at every split is invertible.  A general $A$ need not have this property.  We can, however, enforce it by applying one row permutation from the left. In fact, Gaussian elimination gives a permutation matrix $P_A$ and triangular matrices $L_{\rm s}$ and $U_{\rm s}$ such that
\begin{equation}
 M\coloneqq P_A A=L_{\rm s}U_{\rm s},
 \label{eq:global-scalar-PLU}
\end{equation}
where $L_{\rm s}$ and $U_{\rm s}$ are lower and upper triangular with unit diagonal entries, respectively. 
This one permutation is sufficient for the entire recursion.  To see this, split the triangular factors at any recursive boundary after $p$ coordinates:
\begin{equation}
 L_{\rm s}=\begin{pmatrix}L_{11}&0\\L_{21}&L_{22}\end{pmatrix},
 \qquad
 U_{\rm s}=\begin{pmatrix}U_{11}&U_{12}\\0&U_{22}\end{pmatrix},
\end{equation}
where $L_{11}, U_{11} \in \GL_p(\F_2)$ and $L_{22}, U_{22} \in \GL_{bk'-p}(\F_2)$ because they are triangular matrices with unit diagonal entries. The matrices $L_s$ and $U_s$ can be rewritten as
\begin{align}
    \begin{pmatrix}
        L_{11} & 0 \\
        L_{21} & L_{22}
    \end{pmatrix}
    &= \begin{pmatrix}
        I & 0 \\
        L_{21} L_{11}^{-1} & I
    \end{pmatrix} \begin{pmatrix}
        L_{11} & 0 \\
        0 & L_{22}
    \end{pmatrix}, \\
    \begin{pmatrix}
        U_{11} & U_{12} \\ 0 & U_{22}
    \end{pmatrix}
    &= \begin{pmatrix}
        U_{11} & 0 \\
        0 & U_{22}
    \end{pmatrix} \begin{pmatrix}
        I & U_{11}^{-1} U_{12} \\ 0 & I
    \end{pmatrix},
\end{align}
and thus the top-left quadrant of $M$ is given by $L_{11} U_{11}$, which is invertible.
Furthermore, we can repeat this procedure on the diagonal blocks $L_{ii}, U_{ii}$ for $i \in [2]$ with arbitrary splits into two parts, which recursively splits triangular matrices into diagonal triangular blocks of smaller size with additional cross-half factors coming out.
The argument above guarantees that no further permutation is needed.

It remains to account for the cost of $P_A$, which can be bounded from above as follows.
First, recall that any permutation is a product of two involutions
\begin{align}
    P_A = \tau_S \cdot \tau_T,
\end{align}
and each involution is a collection of disjoint swaps. If $\tau_S$ and $\tau_T$ are composed of swapping whole $k'$-qubit blocks, their implementations are inexpensive. 
This is because swapping two whole blocks can be implemented in depth three by
\begin{equation}
 \begin{pmatrix}0&I_{k'}\\I_{k'}&0\end{pmatrix}
 =U(I_{k'})L(I_{k'})U(I_{k'}), \label{eq:swap_block}
\end{equation}
and thus $P_A$ can be implemented by depth six in this best case.
In general, $\tau_S$ and $\tau_T$ may swap individual coordinates between blocks and within a block. For individual-coordinate swaps between blocks, one needs a partial swap between two code blocks, which requires $k' \times k'$ partial permutation matrices instead of $I_{k'}$ as in Eq.~\eqref{eq:swap_block}.
Let
\begin{equation}
 \rho_{\rm pp}
 =\max_R\max\{\rho_\Gamma(R),\rho_\Gamma(R^T)\}
 \leq\rho_{\max}, \label{eq:overhead_partial_permutation}
\end{equation}
where $R$ ranges over $k'\times k'$ partial permutation matrices, that is, matrices with at most one nonzero entry in each row and column.  Collect the coordinate swaps between blocks $i$ and $j$ into one such matrix $R$.  Their joint action is
\begin{equation}
 \operatorname{Swap}_{ij}(R)
 =T_{ij}(R)T_{ji}(R^T)T_{ij}(R).
 \label{eq:partial-block-swap}
\end{equation}
Indeed, on these two blocks, the product is
\begin{equation}
 \begin{pmatrix}
 I_{k'}+RR^T&R\\
 R^T&I_{k'}+R^T R
 \end{pmatrix},
\end{equation}
which exchanges exactly the coordinates selected by $R$.  Here, we used $RR^T R=R$, which holds for every partial permutation matrix.  One partial
swap thus has depth at most $3(2\rho_{\rm pp}+1)$.

Within either involution, a block can exchange coordinates with at most $\min\{k',b-1\}$ other blocks.  The inter-block exchanges can therefore be divided into at most $\min\{k',b-1\}+1$ rounds of disjoint block pairs, which is the standard edge-coloring bound used in Theorem~9 of Ref.~\cite{Malcolm2026}.  Applying the bound to both involutions gives
\begin{equation}
 2\cdot3(2\rho_{\rm pp}+1)
 \bigl(\min\{k',b-1\}+1\bigr)
\end{equation}
for their inter-block parts.  The remaining swaps within individual blocks have the form $Q_1\oplus\cdots\oplus Q_b$.  Pairing the blocks and using the six-factor synthesis costs $6(2\rho_{\max}+1)$ for each
involution.  Hence, the depth $d(\mathcal{D}(P_A))$ for synthesizing $\mathcal{D}(P_A)$ is upper bounded by
\begin{equation}
 d\bigl(\mathcal D(P_A)\bigr)
 \leq
 (12\rho_{\rm pp}+6)\bigl(\min\{k',b-1\}+1\bigr)
 +24\rho_{\max}+\Theta(1).
 \label{eq:scalar-permutation-cost}
\end{equation}

As a result, for a general singular $A$, we have
\begin{equation}
\begin{split}
 d\bigl(\mathcal D(A)\bigr)
 &\leq d_{\rm base}(b)
 +
 d\bigl(\mathcal D(P_A)\bigr).
\end{split}
 \label{eq:block-pldu-full-cost}
\end{equation}

\subsection{Overall layer counts for a general symplectic matrix element}

We can thus count the total number of layers of gates necessary to implement a general symplectic matrix element Eq.~\eqref{eq:alternative_LULU} and Eq.~\eqref{eq:invertible-A-LDU}. The additional implementation cost for a general decomposition Eq.~\eqref{eq:general-A-LDU} is the same between the two. 
For the LULU decomposition Eq.~\eqref{eq:alternative_LULU}, the total gate depth is
\begin{equation}
    d_{\text{LULU}}=4(2b+34)\rho_{\max} + 12b + \Theta(1), \label{eq:total-cost-LULU-2}
\end{equation}
as shown in Eq.~\eqref{eq:total-cost-LULU}, which follows from Eq.~\eqref{eq:one-symmetric-shear-upper}.
Alternatively, for the LDU decomposition Eq.~\eqref{eq:invertible-A-LDU}, the total depth is 
\begin{equation}
    \begin{split}
    d_{\text{LDU}}&=2(2b+34)\rho_{\max} + 6 b + d_{\rm base}(b) \\
    & \qquad + d\bigl(\mathcal D(P_A)\bigr) + \Theta(1), 
    \end{split}\label{eq:total_cost_LDU}
\end{equation}
which follows from Eq.~\eqref{eq:one-symmetric-shear-upper} and \eqref{eq:block-pldu-full-cost}.
The leading-order term in $d_{\rm base}(b)$ is $4b\rho_{\max}$ from Eqs.~\eqref{eq:block-plu-rounds-general} and \eqref{eq:block-plu-base-cost}, and the term $d\bigl(\mathcal D(P_A)\bigr)$ is, in the worst case, $k'$- or $b$-dependent factor as can be seen from Eqs.~\eqref{eq:scalar-permutation-cost}.
Thus, the former decomposition in general gives a lower implementation overhead for large $b$ and $k'$, unless the D term in the LDU decomposition has a nice form. This conclusion is the same as Ref.~\cite{Malcolm2026}. The difference is that the in-block shear gates are more costly in our model, which leads to a large constant in Eq.~\eqref{eq:total-cost-LULU-2}, in addition to the basic layer overhead $ \rho_{\max}=k'^{2-o(1)}$ for in-block gate implementations instead of $O(k')$ in Ref.~\cite{Malcolm2026}.

When the permutation ${\cal U}(E)$ in Eq.~\eqref{eq:general-A-LDU} is necessary for the case of a singular top-left quadrant, the additional overhead is incurred, which depends heavily on the structure of $E$. If $E$ happens to be composed of inter-block swaps, then the additional cost is at most $2\rho_{\rm pp}+\Theta(1)$ from Eqs.~\eqref{eq:offdiagonal-shear-gadget} and \eqref{eq:overhead_partial_permutation}. If in-block swaps are contained, $\mathcal D_i(Q_s)$ synthesis in Eq.~\eqref{eq:local-phase-conjugation} may be required. From Prop.~\ref{prop:symmetric_matrix_decomp} and the fact that neighboring local diagonal gates can be merged as explained above, this costs at most three times the six-factor decomposition as in Eq.~\eqref{eq:six-block-decomp}, each factor of which requires $2\rho_{\rm pp}+\Theta(1)$ gate gadgets in the worst case. Thus, the overhead should be $3\cdot 6\cdot2\rho_{\rm pp}+\Theta(1)=36\rho_{\rm pp}+\Theta(1)$ in this case.

\section{Small-size examples of Steinberg codes}
In this appendix, we introduce small-size examples of Steinberg codes, which are also used in our numerical study in Sec.~\ref{sec:numerical-study}. 
We also summarize an elementary point of view of Steinberg codes useful for deriving their parity check matrices.

\subsection{Classical linear codes from homology made of building}
Let us consider a chain complex over a simplicial complex $\Delta = \cup_{i=0}^n \Delta_i$ with coefficient over $\F_2$:
\begin{align}
    C_n \xrightarrow[]{\partial_n} C_{n-1} \xrightarrow[]{\partial_{n-1}} \cdots \xrightarrow[]{\partial_1} C_0 \xrightarrow[]{\partial_0} \mathbb{F}_2. 
\end{align}
For a Steinberg code, the classical linear code is associated with the highest level cycle space, i.e., $Z_n \coloneqq \ker{\partial_n}$. Note that $n$ here is related to the Lie rank $r$ by $r=n+1$.
An element $x$ in $C_n$ is a formal sum over $n$-dimensional simplicial complex $c \in \Delta_n$, which is written as
\begin{align}
    x = \sum_{c \in \Delta_n} f_c(x) \cdot c.
\end{align}
A codeword of the associated code is given by $(f_c(z))_{c \in \Delta_n}$ for $z \in Z_n$.

Let $\Delta$ be a flag complex over $\F_2$; i.e., each simplex is an increasing sequence of subspaces over a finite field $\F_q$, and each boundary of a simplex $c$ contains every sequence with one subspace deleted from $c$.
Note that all flag complexes are simplicial complexes. In this complex, an $n$-dimensional simplex $c \in \Delta_n$ is a complete flag
\begin{align}
    c = (V_0 \subset V_1 \subset \cdots \subset V_n).
\end{align}

If a group $G$ acts on $V_n$, it acts on $\Delta$ as follows:
\begin{align}
    g \cdot c = (g \cdot V_0 \subset g \cdot V_1 \subset \cdots \subset g \cdot V_n),
\end{align}
where $g$ is a group element.
$G$ acts flag-transitively on $\Delta$ if there exists $g \in G$ for any $c, c' \in \Delta_n$ such that $c' = g \cdot c$.

Let $G$ be a Chevalley group and $B$ a Borel subgroup. The corresponding spherical building $\Delta$ is a flag complex with the following structure: $B$ stabilizes a specific complete flag $c_{f}$ called the fundamental chamber, and $G$ acts transitively on $\Delta$, which implies that the chambers are $\Delta_n = \{g\cdot c_f: g \in G\}$. Thus, the chambers $\Delta_n$ have one-to-one correspondence with $G/B$.
By Solomon-Tits theorem, the cycle space $Z_n$ has a $G$ representation of dimension $q^N = \dim Z_n$~\cite{SMITH1995199}. 

Since the fundamental chamber is stabilized by the Borel subgroup $B$, its panel can also be described by its stabilizer $P$, but it could change the chamber. This means that for the fundamental chamber $c_f$ with
\begin{align}
    c_f = V_0 \subset V_1 \subset \cdots \subset V_n,
\end{align}
one of its panels $p_i$ can be obtained by removing one of the subspaces $V_i$ as
\begin{align}
    p_i = V_0 \subset \cdots V_{i-1} \subset V_{i+1} \subset \cdots \subset V_n.
\end{align}
The stabilizer $P_i$ of $p_i$ stabilizes $V_j$ for $j \neq i$ but may change $V_i$.  The panels $\{p_0,\ldots,p_n\}$ is called the fundamental panels.
Other panels of the type $i$ is given by $\{g\cdot p_i:g\in G\}$, which thus implies that the panel of the type $i$ has one-to-one correspondence with $\{gP_i, g \in G\}$. If a chamber $c$ that corresponds to $gB$ contains the panel $p_i$ that corresponds to $hP_i$, then $gB \subset hP_i$. This condition is equivalent to $h^{-1}g \in P_i$, which could be a more helpful criterion for deriving parity check matrices. 
\begin{prop}
    Given a chamber $c$ and a panel $p$, the following holds
\begin{align}
\begin{split}
    p \subset c 
    &\iff h^{-1}g \in P_i,\quad \forall g, h,\ c=g\cdot c_f,\ p = h\cdot p_i.
    \label{eq:chamber-panel-criterion}
\end{split}
\end{align}   
\end{prop}
\begin{proof}
$(\Rightarrow)$ Since $B$ and $P_i$ are stabilizers of $c_f$ and $p_i$, respectively, $p\subset c$ implies $g B \subset h P_i$.
Given $gB \subset hP_i$, for all $b \in B$, there exists $q \in P_i$ such that $gb = hq $, which implies $ h^{-1}g = qb^{-1}$. Therefore, $h^{-1}g \in P_i$ since $B \subset P_i$. 

\noindent $(\Leftarrow)$ Assume that $h^{-1}g \in P_i$. Then, there exists $f \in P_i$ such that $h^{-1}g = f $, which implies $ g = hf$. Since $p_i \subset c_f$, we have $p = h\cdot p_i = hf\cdot p_i =g\cdot p_i \subset g\cdot c_f = c$.    
\end{proof}
The parity-check matrix $H$ of $C_\text{St}(G)$ is given by
    \begin{align}
        H_{pc} &= \begin{cases}
        1 & \text{if $p \subset c$} \\
        0 & \text{if $p \not\subset c$} 
    \end{cases}, \\
    &= \begin{cases}
        1 & \text{if $h^{-1}g \in P_
        i$}, \\
        0 & \text{if $h^{-1}g \notin P_i$} 
    \end{cases}.
\end{align}

\subsection{Examples}
\label{appendix:code-examples}

\subsubsection{$[21, 8, 6]$ code from $\mathrm{SL}_3(\F_2)$}
With group $\operatorname{SL}_3(\F_2)$, the code parameters $[21, 8, 6]$ of its Steinberg code is given by the result in Sec.~\ref{par:subfamily}.
Its building corresponds to the 1-dimensional flag complex of the Fano plane~\cite{SMITH1995199, Abramenko2008}.
The panels of this building correspond to 7 points and 7 lines of the Fano plane. A chamber is a full flag, i.e., a pair of a point and a line of the Fano plane. Since each line contains 3 points, there are 21 chambers in total. 

The parity-check matrix $H \in \F_2^{14 \times 21}$ is determined by the incidence relation between chambers and panels and explicitly given by
\begin{equation}
H = \left(
\begin{smallmatrix}
1&1&1&0&0&0&0&0&0&0&0&0&0&0&0&0&0&0&0&0&0\\
1&0&0&1&0&0&0&0&1&0&0&0&0&0&0&0&0&0&0&0&0\\
0&1&0&0&0&0&0&0&0&0&1&0&0&0&0&0&0&0&1&0&0\\
0&0&1&0&0&0&0&0&0&0&0&0&1&0&0&1&0&0&0&0&0\\
0&0&0&1&1&1&0&0&0&0&0&0&0&0&0&0&0&0&0&0&0\\
0&0&0&0&1&0&0&0&0&0&0&0&0&1&0&0&0&0&0&1&0\\
0&0&0&0&0&1&0&0&0&0&0&1&0&0&0&0&0&1&0&0&0\\
0&0&0&0&0&0&1&1&1&0&0&0&0&0&0&0&0&0&0&0&0\\
0&0&0&0&0&0&1&0&0&1&0&0&0&0&1&0&0&0&0&0&0\\
0&0&0&0&0&0&0&1&0&0&0&0&0&0&0&0&1&0&0&0&1\\
0&0&0&0&0&0&0&0&0&1&1&1&0&0&0&0&0&0&0&0&0\\
0&0&0&0&0&0&0&0&0&0&0&0&1&1&1&0&0&0&0&0&0\\
0&0&0&0&0&0&0&0&0&0&0&0&0&0&0&1&1&1&0&0&0\\
0&0&0&0&0&0&0&0&0&0&0&0&0&0&0&0&0&0&1&1&1
\end{smallmatrix}\right).
\end{equation}
The generator matrix $G\in\F_2^{8\times 21}$ is determined by apartments, which are ``cycles'' in a Fano plane. Here, the ``cycle'' is in the sense of a simplicial complex; i.e., point--line pairs share the same points and lines exactly twice for each.
Explicitly, one choice of the generator matrix is given by 
\begin{equation}
    G= \left(\begin{smallmatrix}
        1&1&0&1&1&0&0&0&0&0&0&0&0&0&0&0&0&0&1&1&0\\
 0&0&0&1&1&0&0&1&1&0&0&0&0&0&0&0&0&0&0&1&1\\
 0&0&0&0&0&0&1&1&0&1&1&0&0&0&0&0&0&0&1&0&1\\
 0&0&0&0&0&0&0&0&0&1&1&0&0&1&1&0&0&0&1&1&0\\
 0&0&0&0&0&0&0&0&0&0&0&0&1&1&0&1&1&0&0&1&1\\
 0&1&1&0&0&0&0&0&0&0&0&0&0&0&0&1&1&0&1&0&1\\
 0&0&0&0&1&1&0&0&0&0&0&0&0&0&0&0&1&1&0&1&1\\
 0&0&0&0&1&1&0&0&0&0&1&1&0&0&0&0&0&0&1&1&0
    \end{smallmatrix}\right).
    \label{eq:generator_Fano}
\end{equation}

The Tanner-graph automorphisms for this parity-check matrix are generated by two coordinate permutations given by
\begin{align}
    \begin{split}
    \sigma_1 & \coloneqq (2~3)(7~8)(10~17)(11~16)(12~18)\\
    &\hspace{2cm}(13~19)(14~20)(15~21),
    \end{split}\\
    \begin{split}
    \sigma_2 & \coloneqq (1~19~5)(2~20~4)(3~21~6)(7~10~15)\\
    &\hspace{2cm}(8~12~13)(9~11~14)(16~17~18).
    \end{split}
\end{align}
Their induced action on the logical space is given for a generator matrix Eq.~\eqref{eq:generator_Fano} by
\begin{align}
    \pi_{\rm St}(\sigma_1) &= \begin{pmatrix}
        1&0&0&0&1&1&0&0\\
 0&1&1&1&0&0&0&0\\
 0&0&1&1&1&0&0&0\\
 0&0&0&0&1&0&0&0\\
 0&0&0&1&0&0&0&0\\
 0&0&0&1&1&1&0&0\\
 0&0&0&1&0&0&0&1\\
 0&0&0&0&1&0&1&0
    \end{pmatrix}, \\
    \pi_{\rm St}(\sigma_2) &= \begin{pmatrix}
        1&0&0&0&0&0&0&0\\
 1&0&0&0&0&0&0&1\\
 0&0&0&1&0&0&0&1\\
 0&1&1&1&0&0&0&0\\
 0&1&0&0&0&0&1&0\\
 0&0&0&0&0&0&1&0\\
 1&0&0&0&0&1&1&0\\
 1&0&0&0&1&1&0&0
    \end{pmatrix}.
\end{align}

\subsubsection{$[45, 16, 8]$ code from $\Sp_4(\mathbb{F}_2)$}
We take $G = \mathrm{Sp}_4(\F_2) \cong S_6$, which is also a finite group of Lie type. Considering its building \cite{Abramenko2008}, its dimension 1 simplex consists of points $\{\{i, j\}:i,j\in\{1,\ldots,6\},i\neq j\}$ and a line $\{\{\{i, j\}, \{k, l\}, \{m, n\}\}:i,j,k,l,m,n\in\{1,\ldots,6\},\text{all distinct}\}$. First, we fix a chamber. Let this chamber be $c_f = \{1, 2\} \subset \{\{1, 2\}, \{3, 4\}, \{5, 6\}\}$. The Borel subgroup $B$ stabilizes it and thus given by
\begin{align}
    B = \langle (1~2), (3~4), (5~6), (3~5)(4~6)\rangle.   
\end{align}
Other chambers are generated by the transitive action of $S_6$, and the number of chambers is $[G:B] = |S_6| / |B| = 6! / 2^4 = 45$.

The subgroup $P_p$ of $G$ is defined to stabilize the point $\{1, 2\}$ of $c_f$, and thus given by
\begin{align}
    P_p = \langle (1~2) \rangle \times S_4, 
\end{align}
where $S_4$ permutes $\{3, 4, 5, 6\}$. Hence, the number of this type of panels is $[G : P_p] = |S_6| / |P_p| = 6! / (2 \times |S_4|) = 15$. 
On the other hand, the subgroup $P_l$ of $G$ that stabilizes the line-type panel $\{\{1, 2\}, \{3, 4\}, \{5, 6\}\}$ of $c_f$ is given by
\begin{align}
    P_l = \langle (1~2), (3~4), (5~6), (1~3)(2~4), (3~5)(4~6) \rangle.
\end{align}
Since $|P_l| = 48$, the number of line-type panels is $[G : P_l] = 15$.

The parity-check matrix $H \in \F_2^{30 \times 45}$ of this code is concretely calculated with criterion~\ref{eq:chamber-panel-criterion}, which is given by
\begin{widetext}
\begin{equation}
    H = \left(\begin{smallmatrix}
1 & 0 & 0 & 0 & 0 & 0 & 0 & 0 & 0 & 0 & 0 & 0 & 0 & 0 & 0 & 0 & 0 & 0 & 0 & 1 & 0 & 0 & 0 & 0 & 0 & 0 & 0 & 0 & 0 & 0 & 0 & 0 & 0 & 0 & 0 & 1 & 0 & 0 & 0 & 0 & 0 & 0 & 0 & 0 & 0 \\
0 & 1 & 0 & 0 & 0 & 0 & 0 & 0 & 0 & 0 & 0 & 0 & 0 & 0 & 0 & 1 & 0 & 0 & 0 & 0 & 0 & 0 & 0 & 0 & 0 & 0 & 1 & 0 & 0 & 0 & 0 & 0 & 0 & 0 & 0 & 0 & 0 & 0 & 0 & 0 & 0 & 0 & 0 & 0 & 0 \\
0 & 0 & 1 & 0 & 1 & 0 & 0 & 0 & 0 & 0 & 0 & 0 & 0 & 0 & 0 & 0 & 0 & 0 & 0 & 0 & 0 & 0 & 0 & 0 & 0 & 0 & 0 & 0 & 0 & 0 & 1 & 0 & 0 & 0 & 0 & 0 & 0 & 0 & 0 & 0 & 0 & 0 & 0 & 0 & 0 \\
0 & 0 & 0 & 1 & 0 & 0 & 1 & 0 & 1 & 0 & 0 & 0 & 0 & 0 & 0 & 0 & 0 & 0 & 0 & 0 & 0 & 0 & 0 & 0 & 0 & 0 & 0 & 0 & 0 & 0 & 0 & 0 & 0 & 0 & 0 & 0 & 0 & 0 & 0 & 0 & 0 & 0 & 0 & 0 & 0 \\
0 & 0 & 0 & 0 & 0 & 1 & 0 & 0 & 0 & 0 & 0 & 0 & 0 & 0 & 0 & 0 & 0 & 0 & 0 & 0 & 0 & 0 & 0 & 0 & 0 & 0 & 0 & 0 & 0 & 0 & 0 & 1 & 0 & 0 & 0 & 0 & 0 & 0 & 0 & 0 & 0 & 0 & 0 & 0 & 1 \\
0 & 0 & 0 & 0 & 0 & 0 & 0 & 1 & 0 & 0 & 1 & 0 & 1 & 0 & 0 & 0 & 0 & 0 & 0 & 0 & 0 & 0 & 0 & 0 & 0 & 0 & 0 & 0 & 0 & 0 & 0 & 0 & 0 & 0 & 0 & 0 & 0 & 0 & 0 & 0 & 0 & 0 & 0 & 0 & 0 \\
0 & 0 & 0 & 0 & 0 & 0 & 0 & 0 & 0 & 1 & 0 & 0 & 0 & 0 & 0 & 0 & 0 & 1 & 0 & 0 & 0 & 0 & 0 & 0 & 0 & 0 & 0 & 0 & 0 & 0 & 0 & 0 & 0 & 0 & 0 & 0 & 0 & 0 & 0 & 0 & 0 & 0 & 1 & 0 & 0 \\
0 & 0 & 0 & 0 & 0 & 0 & 0 & 0 & 0 & 0 & 0 & 1 & 0 & 0 & 1 & 0 & 1 & 0 & 0 & 0 & 0 & 0 & 0 & 0 & 0 & 0 & 0 & 0 & 0 & 0 & 0 & 0 & 0 & 0 & 0 & 0 & 0 & 0 & 0 & 0 & 0 & 0 & 0 & 0 & 0 \\
0 & 0 & 0 & 0 & 0 & 0 & 0 & 0 & 0 & 0 & 0 & 0 & 0 & 1 & 0 & 0 & 0 & 0 & 0 & 0 & 0 & 1 & 0 & 0 & 0 & 0 & 0 & 1 & 0 & 0 & 0 & 0 & 0 & 0 & 0 & 0 & 0 & 0 & 0 & 0 & 0 & 0 & 0 & 0 & 0 \\
0 & 0 & 0 & 0 & 0 & 0 & 0 & 0 & 0 & 0 & 0 & 0 & 0 & 0 & 0 & 0 & 0 & 0 & 1 & 0 & 0 & 0 & 0 & 0 & 0 & 1 & 0 & 0 & 0 & 0 & 0 & 0 & 0 & 1 & 0 & 0 & 0 & 0 & 0 & 0 & 0 & 0 & 0 & 0 & 0 \\
0 & 0 & 0 & 0 & 0 & 0 & 0 & 0 & 0 & 0 & 0 & 0 & 0 & 0 & 0 & 0 & 0 & 0 & 0 & 0 & 1 & 0 & 1 & 1 & 0 & 0 & 0 & 0 & 0 & 0 & 0 & 0 & 0 & 0 & 0 & 0 & 0 & 0 & 0 & 0 & 0 & 0 & 0 & 0 & 0 \\
0 & 0 & 0 & 0 & 0 & 0 & 0 & 0 & 0 & 0 & 0 & 0 & 0 & 0 & 0 & 0 & 0 & 0 & 0 & 0 & 0 & 0 & 0 & 0 & 1 & 0 & 0 & 0 & 0 & 1 & 0 & 0 & 0 & 0 & 1 & 0 & 0 & 0 & 0 & 0 & 0 & 0 & 0 & 0 & 0 \\
0 & 0 & 0 & 0 & 0 & 0 & 0 & 0 & 0 & 0 & 0 & 0 & 0 & 0 & 0 & 0 & 0 & 0 & 0 & 0 & 0 & 0 & 0 & 0 & 0 & 0 & 0 & 0 & 1 & 0 & 0 & 0 & 1 & 0 & 0 & 0 & 0 & 0 & 0 & 0 & 0 & 1 & 0 & 0 & 0 \\
0 & 0 & 0 & 0 & 0 & 0 & 0 & 0 & 0 & 0 & 0 & 0 & 0 & 0 & 0 & 0 & 0 & 0 & 0 & 0 & 0 & 0 & 0 & 0 & 0 & 0 & 0 & 0 & 0 & 0 & 0 & 0 & 0 & 0 & 0 & 0 & 1 & 1 & 0 & 0 & 0 & 0 & 0 & 1 & 0 \\
0 & 0 & 0 & 0 & 0 & 0 & 0 & 0 & 0 & 0 & 0 & 0 & 0 & 0 & 0 & 0 & 0 & 0 & 0 & 0 & 0 & 0 & 0 & 0 & 0 & 0 & 0 & 0 & 0 & 0 & 0 & 0 & 0 & 0 & 0 & 0 & 0 & 0 & 1 & 1 & 1 & 0 & 0 & 0 & 0 \\
1 & 0 & 1 & 0 & 0 & 0 & 0 & 1 & 0 & 0 & 0 & 0 & 0 & 0 & 0 & 0 & 0 & 0 & 0 & 0 & 0 & 0 & 0 & 0 & 0 & 0 & 0 & 0 & 0 & 0 & 0 & 0 & 0 & 0 & 0 & 0 & 0 & 0 & 0 & 0 & 0 & 0 & 0 & 0 & 0 \\
0 & 1 & 0 & 1 & 0 & 0 & 0 & 0 & 0 & 0 & 0 & 0 & 0 & 0 & 0 & 0 & 0 & 0 & 0 & 0 & 0 & 0 & 0 & 0 & 1 & 0 & 0 & 0 & 0 & 0 & 0 & 0 & 0 & 0 & 0 & 0 & 0 & 0 & 0 & 0 & 0 & 0 & 0 & 0 & 0 \\
0 & 0 & 0 & 0 & 1 & 0 & 0 & 0 & 0 & 0 & 0 & 0 & 0 & 0 & 0 & 0 & 0 & 0 & 0 & 0 & 0 & 0 & 0 & 0 & 0 & 0 & 0 & 0 & 0 & 1 & 0 & 0 & 1 & 0 & 0 & 0 & 0 & 0 & 0 & 0 & 0 & 0 & 0 & 0 & 0 \\
0 & 0 & 0 & 0 & 0 & 1 & 1 & 0 & 0 & 0 & 0 & 1 & 0 & 0 & 0 & 0 & 0 & 0 & 0 & 0 & 0 & 0 & 0 & 0 & 0 & 0 & 0 & 0 & 0 & 0 & 0 & 0 & 0 & 0 & 0 & 0 & 0 & 0 & 0 & 0 & 0 & 0 & 0 & 0 & 0 \\
0 & 0 & 0 & 0 & 0 & 0 & 0 & 0 & 1 & 0 & 0 & 0 & 0 & 0 & 0 & 0 & 0 & 0 & 0 & 1 & 1 & 0 & 0 & 0 & 0 & 0 & 0 & 0 & 0 & 0 & 0 & 0 & 0 & 0 & 0 & 0 & 0 & 0 & 0 & 0 & 0 & 0 & 0 & 0 & 0 \\
0 & 0 & 0 & 0 & 0 & 0 & 0 & 0 & 0 & 1 & 1 & 0 & 0 & 0 & 0 & 1 & 0 & 0 & 0 & 0 & 0 & 0 & 0 & 0 & 0 & 0 & 0 & 0 & 0 & 0 & 0 & 0 & 0 & 0 & 0 & 0 & 0 & 0 & 0 & 0 & 0 & 0 & 0 & 0 & 0 \\
0 & 0 & 0 & 0 & 0 & 0 & 0 & 0 & 0 & 0 & 0 & 0 & 1 & 0 & 0 & 0 & 0 & 0 & 1 & 0 & 0 & 0 & 0 & 0 & 0 & 0 & 0 & 0 & 0 & 0 & 0 & 1 & 0 & 0 & 0 & 0 & 0 & 0 & 0 & 0 & 0 & 0 & 0 & 0 & 0 \\
0 & 0 & 0 & 0 & 0 & 0 & 0 & 0 & 0 & 0 & 0 & 0 & 0 & 1 & 1 & 0 & 0 & 0 & 0 & 0 & 0 & 0 & 0 & 0 & 0 & 0 & 0 & 0 & 0 & 0 & 1 & 0 & 0 & 0 & 0 & 0 & 0 & 0 & 0 & 0 & 0 & 0 & 0 & 0 & 0 \\
0 & 0 & 0 & 0 & 0 & 0 & 0 & 0 & 0 & 0 & 0 & 0 & 0 & 0 & 0 & 0 & 1 & 1 & 0 & 0 & 0 & 0 & 0 & 0 & 0 & 0 & 0 & 0 & 0 & 0 & 0 & 0 & 0 & 0 & 0 & 0 & 0 & 0 & 0 & 0 & 0 & 0 & 0 & 1 & 0 \\
0 & 0 & 0 & 0 & 0 & 0 & 0 & 0 & 0 & 0 & 0 & 0 & 0 & 0 & 0 & 0 & 0 & 0 & 0 & 0 & 0 & 1 & 1 & 0 & 0 & 1 & 0 & 0 & 0 & 0 & 0 & 0 & 0 & 0 & 0 & 0 & 0 & 0 & 0 & 0 & 0 & 0 & 0 & 0 & 0 \\
0 & 0 & 0 & 0 & 0 & 0 & 0 & 0 & 0 & 0 & 0 & 0 & 0 & 0 & 0 & 0 & 0 & 0 & 0 & 0 & 0 & 0 & 0 & 1 & 0 & 0 & 0 & 0 & 1 & 0 & 0 & 0 & 0 & 0 & 0 & 0 & 0 & 0 & 0 & 0 & 0 & 0 & 1 & 0 & 0 \\
0 & 0 & 0 & 0 & 0 & 0 & 0 & 0 & 0 & 0 & 0 & 0 & 0 & 0 & 0 & 0 & 0 & 0 & 0 & 0 & 0 & 0 & 0 & 0 & 0 & 0 & 1 & 1 & 0 & 0 & 0 & 0 & 0 & 0 & 0 & 0 & 0 & 0 & 0 & 1 & 0 & 0 & 0 & 0 & 0 \\
0 & 0 & 0 & 0 & 0 & 0 & 0 & 0 & 0 & 0 & 0 & 0 & 0 & 0 & 0 & 0 & 0 & 0 & 0 & 0 & 0 & 0 & 0 & 0 & 0 & 0 & 0 & 0 & 0 & 0 & 0 & 0 & 0 & 1 & 1 & 0 & 0 & 1 & 0 & 0 & 0 & 0 & 0 & 0 & 0 \\
0 & 0 & 0 & 0 & 0 & 0 & 0 & 0 & 0 & 0 & 0 & 0 & 0 & 0 & 0 & 0 & 0 & 0 & 0 & 0 & 0 & 0 & 0 & 0 & 0 & 0 & 0 & 0 & 0 & 0 & 0 & 0 & 0 & 0 & 0 & 1 & 1 & 0 & 1 & 0 & 0 & 0 & 0 & 0 & 0 \\
0 & 0 & 0 & 0 & 0 & 0 & 0 & 0 & 0 & 0 & 0 & 0 & 0 & 0 & 0 & 0 & 0 & 0 & 0 & 0 & 0 & 0 & 0 & 0 & 0 & 0 & 0 & 0 & 0 & 0 & 0 & 0 & 0 & 0 & 0 & 0 & 0 & 0 & 0 & 0 & 1 & 1 & 0 & 0 & 1
\end{smallmatrix}\right).
\end{equation}
The generator matrix $G\in\F_2^{16\times 45}$ for this parity-check matrix is given by 
\begin{equation}
    G=\left(
\begin{smallmatrix}
1&1&0&1&0&0&0&1&1&0&1&0&0&0&0&1&0&0&0&1&0&0&0&0&0&0&0&0&0&0&0&0&0&0&0&0&0&0&0&0&0&0&0&0&0\\
1&0&0&0&0&1&1&1&1&0&0&0&1&0&0&0&0&0&0&1&0&0&0&0&0&0&0&0&0&0&0&1&0&0&0&0&0&0&0&0&0&0&0&0&0\\
1&0&1&0&0&0&1&0&1&0&0&1&0&0&1&0&0&0&0&1&0&0&0&0&0&0&0&0&0&0&1&0&0&0&0&0&0&0&0&0&0&0&0&0&0\\
1&0&1&0&0&0&0&0&0&0&0&0&0&0&1&0&1&0&0&0&0&0&0&0&0&0&0&0&0&0&1&0&0&0&0&1&1&0&0&0&0&0&0&1&0\\
1&0&0&0&0&0&0&1&0&1&1&0&0&0&0&0&0&1&0&0&0&0&0&0&0&0&0&0&0&0&0&0&0&0&0&1&1&0&0&0&0&0&0&1&0\\
1&0&1&0&0&0&0&0&0&0&0&0&0&1&0&0&0&0&0&1&1&1&1&0&0&0&0&0&0&0&1&0&0&0&0&0&0&0&0&0&0&0&0&0&0\\
1&0&1&1&1&0&0&0&1&0&0&0&0&0&0&0&0&0&0&1&0&0&0&0&1&0&0&0&0&1&0&0&0&0&0&0&0&0&0&0&0&0&0&0&0\\
1&0&0&0&0&0&0&1&0&0&0&0&1&0&0&0&0&0&1&1&1&0&1&0&0&1&0&0&0&0&0&0&0&0&0&0&0&0&0&0&0&0&0&0&0\\
1&0&0&0&0&0&0&1&0&0&1&0&0&0&0&1&0&0&0&0&0&0&0&0&0&0&1&0&0&0&0&0&0&0&0&1&0&0&1&1&0&0&0&0&0\\
1&0&1&0&0&0&0&0&0&0&0&0&0&1&0&0&0&0&0&0&0&0&0&0&0&0&0&1&0&0&1&0&0&0&0&1&0&0&1&1&0&0&0&0&0\\
1&0&1&0&1&0&0&0&0&0&0&0&0&0&0&0&0&0&0&1&1&0&0&1&0&0&0&0&1&0&0&0&1&0&0&0&0&0&0&0&0&0&0&0&0\\
1&0&0&0&0&0&0&1&0&0&0&0&1&0&0&0&0&0&1&0&0&0&0&0&0&0&0&0&0&0&0&0&0&1&0&1&1&1&0&0&0&0&0&0&0\\
1&0&1&0&1&0&0&0&0&0&0&0&0&0&0&0&0&0&0&0&0&0&0&0&0&0&0&0&0&1&0&0&0&0&1&1&1&1&0&0&0&0&0&0&0\\
1&0&1&0&1&0&0&0&0&0&0&0&0&0&0&0&0&0&0&0&0&0&0&0&0&0&0&0&0&0&0&0&1&0&0&1&0&0&1&0&1&1&0&0&0\\
1&0&0&0&0&0&0&1&0&1&1&0&0&0&0&0&0&0&0&1&1&0&0&1&0&0&0&0&0&0&0&0&0&0&0&0&0&0&0&0&0&0&1&0&0\\
1&0&0&0&0&0&0&1&0&0&0&0&1&0&0&0&0&0&0&0&0&0&0&0&0&0&0&0&0&0&0&1&0&0&0&1&0&0&1&0&1&0&0&0&1
\end{smallmatrix}
\right).
\end{equation}
The Tanner-graph automorphisms for this code can be generated by the following two permutations:
\begin{align}
    \begin{split}
    \sigma_1 &\coloneqq 
(1~16)(2~20)(3~10)(4~9)(5~43)(8~11)(14~44)(15~17)(18~31)(21~25) \\
 &\hspace{3cm} (22~38)(23~35)(24~30)(26~34)(27~36)(28~37)(29~33)(39~40), 
    \end{split}\\
    \begin{split}
    \sigma_2 &\coloneqq (1~20~21~23~26~34~35~30~5~3)
     (2~42~15~11~39~7~43~28~32~44)(4~29~14~13~37) \\&\hspace{2cm}(6~18~27~45~17~16~41~12~10~40)(8~36~9~24~22~19~38~25~33~31).
     \end{split}
\end{align}
\end{widetext}

\end{document}